\documentclass[sigconf,twocolumn,nonacm,dvipsnames]{acmart}

\usepackage{fancyhdr}
\usepackage{appendix}
\usepackage{paralist}
\usepackage{subcaption}
\usepackage{standalone}
\usepackage{color, colortbl}
\usepackage{listings}
\usepackage{scalerel}
\usepackage{tikz}
\usetikzlibrary{trees}
\usepackage{amsmath}
\usepackage{pifont}
\usepackage[framemethod=TikZ]{mdframed}
\usepackage[capitalize,noabbrev]{cleveref}
\usepackage{enumitem}
\usepackage{xspace}

\usepackage{multicol}
\usepackage{multirow}
\usepackage{array}
\usepackage{rotating}
\usepackage{makecell}
\usepackage{xspace}
\definecolor{vlightgray}{gray}{0.85}

\newcommand{\coverage}{\ensuremath{\textsc{Coverage}}}

\usepackage{arydshln}

\newcommand{\fairnessconstraints}{{\ensuremath{\mathcal F}}}
\newcommand{\coverageconstraints}{{\ensuremath{\mathcal C}}}
\newcommand{\mutable}{{\ensuremath{\mathbf M}}}
\newcommand{\immutable}{{\ensuremath{\mathbf I}}}
\newcommand{\patterngroup}{\ensuremath{\pattern_{{\tt grp}}}}
\newcommand{\patterninterv}{\ensuremath{\pattern_{{\tt int}}}}

\Crefname{algocf}{Algorithm}{Algorithms}
\usepackage[colorinlistoftodos]{todonotes}

\newcommand{\paratitle}[1]{
\noindent{\bf #1.}}

\newcommand{\xmark}{\ding{55}}%

\usepackage[vlined,commentsnumbered,ruled]{algorithm2e}
\let\oldnl\nl
\newcommand{\nonl}{\renewcommand{\nl}{\let\nl\oldnl}}

\SetCommentSty{mycommfont}

\def\HiLi{\leavevmode\rlap{\hbox to \hsize{\color{red!20}\leaders\hrule height .8\baselineskip depth .5ex\hfill}}}

\newcommand\independent{\protect\mathpalette{\protect\independenT}{\perp}}
\def\independenT#1#2{\mathrel{\rlap{$#1#2$}\mkern2mu{#1#2}}}

\usepackage{amsfonts}
\usepackage{amsmath}
\usepackage{comment}
\usepackage{multirow}
\usepackage{paralist}
\usepackage{paralist}
\usepackage{mathtools}
\usepackage{framed}

\newtheorem{definition}{Definition}[section]

\newtheorem{example}{Example}[section]

\newcommand{\sr}[1]{\textcolor{purple}{\bf (SR:)~[#1]}{\typeout{#1}}}

\newcommand{\ignore}[1]{}

\newcommand{\red}[1]{{\color{red} {#1}}}

\newcommand{\newtextold}[1]{{#1}}

\newcommand{\cut}[1]{}

\newcommand{\probName}{Prescription Ruleset Selection}

\newcommand{\sysName}{\textsc{FairCap}\xspace}
\newcommand{\attrset}{\ensuremath{\mathbb{A}}}

\newcommand{\dom}{{\tt dom}}

\newcommand{\doop}{{\tt do}}
\newcommand{\db}{\ensuremath{D}}

\newcommand{\exo}{\ensuremath{\mathbf{U}}}
\newcommand{\pattern}{\ensuremath{\mathcal{P}}}
\newcommand{\model}{\ensuremath{\mathcal{G}_\db}}
\newcommand{\edvar}{\ensuremath{\mathbf{V}}}

\definecolor{moonstoneblue}{rgb}{0.45, 0.66, 0.76}
\definecolor{oldlace}{rgb}{0.99, 0.96, 0.9}
\definecolor{mintcream}{rgb}{0.96, 1.0, 0.98}
\definecolor{mintgreen}{rgb}{0.0, 0.5, 0.0}
\definecolor{mistyrose}{rgb}{1.0, 0.89, 0.88}
\definecolor{palegold}{rgb}{0.9, 0.75, 0.54}
\definecolor{palechestnut}{rgb}{0.87, 0.68, 0.69}

\newcommand{\reva}[1]{{\leavevmode\color{black}{#1}}}
\newcommand{\revb}[1]{{\leavevmode\color{black}{#1}}}
\newcommand{\revc}[1]{{\leavevmode\color{black}{#1}}}
\newcommand{\common}[1]{{\leavevmode\color{black}{#1}}}

\def\HiLiG{\leavevmode\rlap{\hbox to \hsize{\color{green!30}\leaders\hrule height .8\baselineskip depth .5ex\hfill}}}
\def\HiLiY{\leavevmode\rlap{\hbox to \hsize{\color{yellow!50}\leaders\hrule height .8\baselineskip depth .5ex\hfill}}}
\usepackage{wrapfig}
\definecolor{light-gray}{gray}{0.95}
\usepackage{tcolorbox}

\usepackage{tikz}

\tcbuselibrary{breakable}

\definecolor{darkgreen}{RGB}{0, 200, 0}
\definecolor{experimentblue}{RGB}{127, 159, 186}
\definecolor{experimentgreen}{RGB}{127, 186, 130}
\definecolor{experimentred}{RGB}{186, 138, 127}
\definecolor{experimentgray}{RGB}{94,94,94}

\newtcolorbox[auto counter,number within=subsection]{ruleset}[1]{
  breakable, boxrule=0pt, boxsep=3pt,left=0pt,right=0pt,top=0pt,bottom=0pt,  colback=white, colframe=black,colbacktitle=experimentgray, title=#1:,
}

\newtcolorbox[auto counter,number within=subsection]{summary}[1]{
  breakable, boxrule=0pt, boxsep=3pt,left=0pt,right=0pt,top=0pt,bottom=0pt, colbacktitle=experimentgreen, title=#1:,
}

\def\independenT#1#2{\mathrel{\rlap{$#1#2$}\mkern2mu{#1#2}}}

\begin{document}

\title{Fair and Actionable Causal Prescription Ruleset}

\author{Benton Li}
\affiliation{%
  \institution{Cornell University}
  \country{USA}
}
\email{cl2597@cornell.edu}

\author{Nativ Levy}
\affiliation{%
  \institution{Technion}
   \country{Israel}
}
\email{nativlevymail@gmail.com}

\author{Brit Youngmann}
\affiliation{%
  \institution{Technion}
   \country{Israel}
}
\email{brity@technion.ac.il}

\author{Sainyam Galhotra}
\affiliation{%
  \institution{Cornell University}
   \country{USA}
}
\email{sg@cs.cornell.edu}

\author{Sudeepa Roy}
\affiliation{%
  \institution{Duke University}
   \country{USA}
}
\email{sudeepa@cs.duke.edu}

\renewcommand{\shortauthors}{Li et al.}

\begin{abstract}

Prescriptions, or actionable recommendations, are commonly generated across various fields to influence key outcomes such as improving public health, enhancing economic policies, or increasing business efficiency. While traditional association-based methods may identify correlations, they often fail to reveal the underlying causal factors needed for informed decision-making. 
On the other hand, in decision making for tasks with significant societal or economic impact, it is crucial to provide recommendations that are justifiable and equitable in terms of the outcome for both the protected and non-protected groups. 
Motivated by these two goals,  this paper introduces a fairness-aware framework leveraging causal reasoning for generating a set of actionable prescription rules (ruleset) toward betterment of an outcome while preventing exacerbating inequalities for protected groups.
By considering group and individual fairness metrics from the literature, we ensure that both protected and non-protected groups benefit from these recommendations, providing a balanced and equitable approach to decision-making. We employ efficient optimizations to explore the vast and complex search space considering both fairness and coverage of the ruleset. Empirical evaluation and case study on real-world datasets 
demonstrates the utility 
of our framework for different use cases.
\end{abstract}

\maketitle




\section{Introduction}
\revc{Prescriptions, or actionable recommendations, are commonly generated across various fields to influence key outcomes such as improving product satisfaction, enhancing economic policies, or increasing business efficiency. 
Policymakers in government, decision-makers in businesses, and data scientists in various fields, often rely on data-driven approaches to identify 
potential actions to influence an outcome of interest, such as increasing income levels or loan approval rates}.
%
While association or prediction-based methods are extensively used in practice to draw useful insights from data, they typically identify correlations among variables and may fail to reveal the underlying causal factors, i.e., which actions may result in an improved outcome, needed for informed decision-making. 

\emph{Causal analysis} or {\em causal inference}, therefore, is considered one of the most important requirements to generate prescriptions that are {\em actionable} and aligned with human reasoning~\cite{imbens2024causal}. Causal inference, and in particular {\em observational studies} for causal inference on collected data (when controlled trials are impossible due to cost or ethical reasons), have been extensively studied in the statistics and artificial intelligence (AI) literature for several decades \cite{rubin2005causal, pearl2009causal}. Motivated by this foundational work on causal inference, the notion of causality has also influenced the field of database research. The causal models from AI have been extended to relational databases \cite{salimi2020causal},  and causality has been incorporated into various data management tasks such as finding responsibilities of inputs toward query answers ~\cite{meliou2010causality, meliou2009so, meliou2014causality}, explanations for query answers \cite{roy2014formal, DBLP:journals/pacmmod/YoungmannCGR24}, data discovery~\cite{galhotra2023metam,youngmann2023causal}, data cleaning~\cite{pirhadi2024otclean,salimi2019interventional}, hypothetical reasoning \cite{galhotra2022causal}, and large system diagnostics~\cite{markakis2024sawmill,causalsim,sage, gudmundsdottir2017demonstration}.

\revc{If-then rules are generally considered interpretable by humans~\cite{lakkaraju2016interpretable,guidotti2018local,van2021evaluating,pradhan2022interpretable,chen2018optimization}.
We give a concrete example of the difference between association and causation in generating prescriptions or recommended actions in the form of if-then rules below}:
\begin{example}	%
\label{example:ex1} {\bf Importance of causal prescriptions:}
Consider the Stack Overflow (SO) annual developer survey
\cite{stackoverflowreport}, where respondents from around the world answer
questions about their jobs and demographics. A sample of the dataset \reva{with a subset of the
attributes (there are 20 attributes)} is presented in \cref{tab:data}.
Alice, a researcher in the United Nations (UN) finance department, is interested in discovering ways to increase the salaries of high-tech employees worldwide. She is looking for a set of actionable recommendations 
to raise the overall average salary.
Using association-based approaches~\cite{chen2018optimization,lakkaraju2016interpretable}, she may discover that individuals residing in the US who identify as straight or heterosexual tend to earn higher salaries (see \cref{exp:quality} for full details). However, this observation merely indicates a correlation: people living in the US, for example, generally earn more than those outside the country. Their comparatively higher salaries are primarily attributable to the country's economy and are unrelated to their sexual orientation. Thus, this observation cannot be used as a prescription rule to increase salary. 
Our causal analysis, on the other hand, reveals that individuals aged 25-34 with dependents would benefit from working as front-end developers.
This results in a \$44,009 annual salary increase on average. \reva{Another observation is that students should pursue an
undergraduate major in CS. 
This can boost their salary by \$22,174 per year} (see details in \cref{sec:casestudy}).
\end{example}





\cut{
In this work, {we study the problem of generating causal insights (referred to as \emph{prescription rules}), which serve as actionable recommendations} to improve an outcome of interest.
Recent works have introduced causality to the field of database research~\cite{meliou2010causality,  meliou2014causality,salimi2020causal,10.14778/3554821.3554902}. It has been incorporated into various tasks including data discovery~\cite{galhotra2023metam,youngmann2023causal}, data cleaning~\cite{pirhadi2024otclean,salimi2019interventional}, and large system diagnostics~\cite{markakis2024sawmill,causalsim,sage, gudmundsdottir2017demonstration}. 
We propose using causal inference to generate prescription rules that are both actionable and justifiable.
}

While generating prescriptions based on causal inference may help in robust decision-making, causal prescriptions that solely consider the betterment of an outcome (like salary) are not enough in practice. 
It is well-known that decision-making in many high-stake applications (like hiring policy, or policy for approving loans by banks) may lead to disparate societal or economic impact on different sub-populations. 
As a shocking example from a recent work called 
CauSumX~\cite{DBLP:journals/pacmmod/YoungmannCGR24} that generates a set of causal explanations for an aggregated view, the explanations generated 
suggest that male individuals do a Bachelor's degree to increase their salary while 
being an unmarried woman 
has the most adverse effect on salary
(borrowed directly 
from Fig.~19 in~\cite{youngmann2024summarizedcausalexplanationsaggregate}). 
We explored this further in the context of generating prescriptions and observed that prescriptions that are not fairness-aware can generate unfair outcomes to some subpopulations which we refer to as the {\em protected group}. Examples include women, Black, Latino, or Native Americans, individuals with a disability, countries with a weaker economy, or other protected groups specific to an application. 

\begin{table*}[]
\footnotesize
    \centering
    	\caption{\textnormal{A subset of the Stack Overflow dataset.}}
         \label{tab:data}
  			\begin{tabular}[b]{|l|l|l|c|l|l|c|l|c|}
  			
				\hline

				\textbf{ID}
    
    
    &\textbf{Gender} &\textbf{Ethnicity}&
				\textbf{Age} &\textbf{Role} &
				\textbf{Education} &\textbf{Country}&\textbf{Undergrad Major}&\textbf{Salary}
				\\ \hline

				1 &Male&White&26&Data Scientist & PhD& US&Computer Science&180k\\
    		2 &Non-binary&White&32&QA developer & Bachelor's degree& US&Mechanical Eng.&83k\\

 3 &Male&South Asian&29&C-suite executive  & Bachelor's degree & India&Computer Science&24k\\


  4 &Female&East Asian&21&Back-end developer & Bachelor's degree & China&Computer Science&19k\\

    \hline
			\end{tabular}
            \vspace{-5mm}
\end{table*}

\begin{example}	%
\label{example:ex2}
{\bf Importance of fair prescriptions:}
Continuing Example~\ref{example:ex1}, while those causal prescription rules are highly beneficial for the overall population, they are considerably less effective for individuals residing in countries with a low GDP (indicating a weaker economy). For this group, the average expected increase in salary is only approximately \$13,000 per year (in contrast to \$44,009 for the entire group). 
Consequently, implementing these rules would exacerbate the disparity between those living in countries with strong economies and those in countries with weaker economies.
\end{example}


The example above shows that focusing solely on maximizing utility (\revc{i.e., increasing income}) can result in a scenario where only some of the population receive significant improvement, while others experience no benefit (\revc{only a small benefit for individuals from countries with weaker economies in our example}). Additionally, even if a large portion of the population receives recommendations, a protected subpopulation might not share the benefits and, worse, their situation could deteriorate, exacerbating inequalities.

Examples~\ref{example:ex1} and \ref{example:ex2} show that it is crucial to provide recommendations that are (1) {\em causal} for the outcome (beyond associations),  and (2) also {\em fair or equitable} in terms of the outcome for both the protected and non-protected groups. While recent work in database research
has focused on deriving {\em causal explanations} for individual data points, aggregated view, or entire datasets~\cite{salimi2018bias,youngmann2022explaining,ma2023xinsight, DBLP:journals/pacmmod/YoungmannCGR24}, and in particular \cite{DBLP:journals/pacmmod/YoungmannCGR24} has considered generating a set of causal explanations for an aggregated view that resemble a ruleset, 
the absence of fairness considerations in generating these causal explanations can lead to unfair outcomes for the protected group.


\medskip
\noindent
\textbf{Our contributions.~} 
Motivated by the dual goals of generating causal and fair prescriptions for the betterment of an outcome, we introduce a {\em fairness-aware framework leveraging causal reasoning for generating a set of actionable prescription rules (ruleset)} called \sysName\ (\underline{Fair} \underline{CA}usal \underline{P}rescription).
Following research on fairness in data management~\cite{stoyanovich2020responsible,galhotra2022causal}, we assume the existence of a \emph{protected subpopulation}, defined by an attribute such as gender or race for people, or GDP of a country. Motivated by the causal explanation rules for an aggregated view \cite{DBLP:journals/pacmmod/YoungmannCGR24}, each prescription rule in our ruleset applies to a sub-population defined by a {\em grouping attribute}, and prescribes a {\em treatment or intervention} to improve the {\em outcome} for this sub-population. Fairness constraints ensure that the expected utility of the protected population is {\em comparable} to the utility of the unprotected individuals. We borrow the notions of \emph{group and individual fairness} from the fairness literature but tailor them for prescription rules. In addition to the fairness constraints, our coverage constraints ensure that a substantial fraction of the population and protected subpopulation receives at least one recommendation. 

\begin{example}
\label{ex:intro_example_3}
Continuing Examples~\ref{example:ex1} and \ref{example:ex2}, Alice uses our proposed system, called \sysName, to impose fairness and coverage constraints to discover useful and equitable recommendations for increasing salaries worldwide. In particular,
Alice chooses to implement a coverage constraint to ensure that the selected rules apply to a significant portion of people worldwide, including a sufficiently large number of individuals from countries with low GDP (the protected group). She also imposes a fairness constraint to ensure that the expected gains for both protected and non-protected groups are comparable.
\reva{She discovers, for example, that for individuals with 6-8 years of coding experience (a subpopulation comprising 21\% of the entire dataset and 25\% of the protected group), pursuing a bachelor’s degree in computer science will increase the expected salary by $\$14.9k$ for protected and by $\$17.8k$ for non-protected}. (See \cref{sec:casestudy} for more details.) This prescription rule applies to a large portion of the population and ensures fairness by providing a similar expected gain for both protected and non-protected groups, and the allowed difference of outcomes between these two populations may be adjusted by choosing appropriate thresholds in the fairness definitions. 
\end{example}

\noindent
Our main contributions are as follows. \\
{\bf (1)} We {\bf develop a framework that generates a set of prescription rules to enhance an outcome of interest (Section~\ref{sec:problem})}. A prescription rule consists of a \emph{grouping pattern} and an \emph{intervention pattern}, representing the target subpopulation and the actionable recommendation for that group, respectively. The strength of the {\em conditional causal effect} (Section~\ref{sec:background-causal}) of this intervention on the subgroup is used to measure the expected utility of a rule. Our objective is to identify the smallest set of rules that maximizes overall expected utility. We refer to this problem as the {\em \probName} problem.
We adopt several notions of fairness (individual vs. group, statistical parity vs. bounded group loss) from the literature to define the {\bf fairness constraints} for our problem. In addition, {\bf coverage constraints} (for individual rules or for a group) ensure that the solution for the \probName\ problem is applied to a sufficient number of individuals and to minimize inequalities. We show NP-hardness for different variants of the problems and properties (matroid) useful in our algorithms. 
\smallskip
    \par
    \noindent
{\bf (2)} We {\bf develop a general three-step algorithm named \sysName to solve the optimization problem of selecting a fair prescription ruleset (Section~\ref{sec:algo})}. The first step involves mining frequent grouping patterns using the Apriori algorithm~\cite{agrawal1994fast}. In the second step, we employ a lattice-based algorithm to find high utility and fair intervention patterns for grouping patterns identified in the previous step. Finally, the third step applies a greedy approach to determine a solution. \sysName\ can be easily adapted to accommodate all variants of the \probName\ problem.

\smallskip
\par
\noindent
{\bf (3) We provide a detailed  case study  (Section~\ref{sec:casestudy}) and experimental analysis (Section~\ref{sec:experiments}) to evaluate our framework and algorithms.}
The case study shows the qualitative difference of different variants of our problem for different choices of the fairness and coverage constraints. The experiments include two datasets, three baselines, and 18 variations of our problem with different constraints. Our evaluations suggest that fairness may come at the cost of expected
utility for everyone. However, without fairness constraints, we often observe a significant disparity between the protected and non-protected. We also observe that
achieving individual fairness is harder than group fairness,
as most high-utility or high-coverage rules are unfair. Lastly, we show that \sysName\ can generate  prescription rules over large datasets in a reasonable time. 


We discuss related work in \cref{sec:related}, review background on causal inference in \Cref{sec:background-causal}, 
and discuss the limitations of our framework and future work in \cref{sec:conc}.

\section{Related work}
\label{sec:related}

Table \ref{tab:relatedwork} outlines the key distinctions between \sysName\ and prior work. The columns in bold emphasize our key contributions: we generate \textbf{causality-based} prescription rules aimed at {improving outcomes} for the \textbf{entire datasets}, while also \textbf{considering fairness}. In contrast, other approaches either produce non-causal rules (as shown in the Causal column), target only subsets of the data (as indicated in the Entire dataset column), or disregard fairness considerations (as highlighted in the Fairness column).

\begin{table}[t]
\small
\centering
\caption{Positioning of our framework w.r.t. previous work.}
\label{tab:relatedwork}
\footnotesize{
\begin{tabular}{|p{25mm}|c|c|c|c|c|c|c|}
    \hline
 
\multicolumn{2}{|c|}{{Related Work}}  & {Causal} &   {\begin{tabular}{@{}c@{}}Fairness\end{tabular}} & {\begin{tabular}{@{}c@{}} Entire\\dataset\end{tabular}} \\
    \hline
\multirow{4}{*}{\begin{tabular}{@{}l@{}}Aggregate Query\\ Result Explanation\end{tabular}}&\cite{ma2023xinsight}&\textbf{\checkmark}&\xmark&\xmark\\
&\cite{salimi2018bias, youngmann2022explaining,DBLP:journals/pacmmod/YoungmannCGR24,roy2014formal,meliou2009so}&\textbf{\checkmark}&\xmark&\checkmark\\

&\cite{li2021putting,miao2019going}&\xmark&\xmark&\xmark\\

&\cite{wu2013scorpion}&\xmark&\xmark&\checkmark\\

\hline
\multirow{2}{*}{\begin{tabular}{@{}l@{}}Interpretable\\ Prediction Models\end{tabular}} 
& \cite{chen2018optimization,lakkaraju2016interpretable}&\xmark&\xmark&{\bf \checkmark}\\
&&&&\\
\hline
\multirow{2}{*}{\begin{tabular}{@{}l@{}}Multi dimensional\\ data aggregation\end{tabular}} &\cite{pastor2021looking,surve2024example,pastor2023hierarchical}&\xmark&\checkmark&\checkmark\\

&\cite{el2014interpretable,kim2020summarizing}&\xmark&\xmark&\checkmark\\

    \hline
    \multicolumn{2}{|c|}{\sysName}
\textbf{}&\textbf{\checkmark}&\textbf{\checkmark}&\textbf{\checkmark}\\
\hline
\end{tabular}
}
\end{table}

\paragraph*{Rule mining}
Association rule mining has been extensively studied \cite{kumbhare2014overview} and is used to identify relationships between items that frequently co-occurring in datasets. These techniques are applied across various fields, such as data analysis and outcome improvement. Notable algorithms include STEM~\cite{girotra2013comparative}, FP-Growth~\cite{kaur2013performance}, AIS~\cite{zhao2003association}, and the Apriori algorithm~\cite{agrawal1994fast}. We leverage the Apriori algorithm to identify sufficiently large subpopulations for which we will generate causal interventions.
Rule-based interpretable prediction models~\cite{yang2017scalable,chen2018optimization,lakkaraju2016interpretable} often leverage association rule mining to generate predictive rules~\cite{lakkaraju2016interpretable,lawless2023interpretable}, with the goal of balancing high predictive accuracy with interpretability~\cite{sagi2021approximating,schielzeth2010simple,lou2013accurate,kim2014bayesian,lawless2023interpretable}. 

Recent work has proposed generating rules based on causal relationships.
In \cite{plecko2022causal}, a framework was introduced to address biases in the data for fair causal analysis. \reva{Other studies have explored the integration of fairness criteria into causal analysis \cite{plecko2023causal,zhang2022adaptive}.}  \cite{sun2021treatment} proposed a method to optimally allocate treatments with uncertain costs that vary based on confounders. Related research has also focused on estimating heterogeneous treatment effects~\cite{wager2018estimation, xie2012estimating,wang2022causal}. However, these approaches assume both treatment and outcome variables are known. In contrast, we assume only the outcome variable is provided and aim to identify treatments that influence the outcome for different subpopulations, potentially yielding different treatments for each group. Our approach ensures the rules apply broadly while maintaining fairness for minority groups.

We adapt the method proposed in \cite{DBLP:journals/pacmmod/YoungmannCGR24} called CauSumX. CauSumX is designed to identify the treatment with the highest causal effect on the outcome for a given subpopulation, generating causal explanations for aggregate queries. A main difference is that CauSumX does not consider fairness. We empirically show that using CauSumX to generate prescription rules can lead to significant disparities between protected and non-protected populations (See \cref{exp:quality}).
Another main difference is that CauSumX considers the \emph{aggregate view} to generate explanations, whereas we consider the entire data, thus, the search space is significantly larger.
Our primary contribution lies in introducing fairness constraints on the generated rules, making the necessary adjustments to the algorithm to scale, and conducting an experimental study to demonstrate the importance of fairness constraints.




\paragraph*{Fairness in data management}
\reva{Algorithmic fairness, especially in the context of predictions by ML algorithms for high-stake decision making, has been a prominent topic in ML and AI (e.g., \cite{agarwal2019fair,mehrabi2021survey,pessach2022review,caton2024fairness,roh2020fairbatch,DBLP:conf/icml/Roh0WS23,jain2024algorithmic,smith2022recsys,ekstrand2019fairness,gao2020counteracting}).} Popular notions of fairness include group and individual fairness~\cite{binns2020apparent,garcia2021maxmin,DBLP:journals/pacmmod/ZhangTPCW23,somerstep2024algorithmic}. Group fairness (measured as statistical parity or equalized odds) ensures that the decision-making process is fair to the protected group but may be unfair towards any specific individual. In contrast, individual notions of fairness enforce that the decisions are fair towards every individual. We refer the reader to Section~\ref{subsec:fairness_constraint} for more details. 
In recent years, fairness has emerged as a key consideration in
data management research~\cite{stoyanovich2020responsible,galhotra2022causal,DBLP:journals/pvldb/TaeZPRW24,DBLP:journals/pacmmod/ZhangTPCW23,jain2024algorithmic}. This includes ensuring fairness during data acquisition~\cite{asudeh2019assessing,nargesian2021tailoring,Nargesian2022}, improving data cleaning processes to promote fairness~\cite{guha2024automated,salimi2019interventional}, and achieving fairness in ranking and in database queries~\cite{zehlike2022fairness,zehlike2022fairness2,li2023query}. \reva{Techniques to ensure fairness of allocated resources~\cite{majumdar2024carma, ehyaei2023robustness} can also be extended to study scenarios where the available resources may be diverse and constrained.
One of our contributions is the introduction of novel definitions of group and individual fairness for causal analysis. We defer the extension of our work to other definitions of fairness to future work.}

\paragraph*{Causal inference in data management} \newtextold{Causality, used as a generic term of cause-effect analysis, has been used in different contexts 
in data management research
}~\cite{meliou2009so, meliou2010complexity,salimi2020causal,10.14778/3554821.3554902}. 
This includes data discovery~\cite{galhotra2023metam,youngmann2023causal,gong2024nexus,santos2021correlation,huang2023fast}, data cleaning~\cite{pirhadi2024otclean,salimi2019interventional}, query result explanation~\cite{salimi2018bias,DBLP:journals/pacmmod/YoungmannCGR24,youngmann2022explaining,ma2023xinsight, roy2014formal}, hypothetical reasoning~\cite{galhotra2022hyper}, and  system diagnostics~\cite{markakis2024sawmill,causalsim,gudmundsdottir2017demonstration}. \newtextold{We use the interventional notion of causal inference on observational data from AI and Statistics \cite{pearl2009causal, rubin2005causal} to define our prescription rules (more in Sections~\ref{sec:background-causal} and \ref{subsec:utility}). } 
to design interventions that improve an outcome of interest.


\paragraph*{Aggregate query result explanation}
A substantial line of research has focused on using provenance to explain aggregate SQL query results~\cite{bidoit2014query,chapman2009not,meliou2010complexity,meliou2009so,DeutchFG20,lee2020approximate,ten2015high,li2021putting}. Other explanation methods include (non-causal) interventions \cite{wu2013scorpion,roy2014formal,roy2015explaining,tao2022dpxplain,DBLP:journals/pvldb/DeutchGMMS22}, and counterbalancing patterns \cite{miao2019going}. Recent work \cite{salimi2018bias,youngmann2022explaining,DBLP:journals/pacmmod/YoungmannCGR24,ma2023xinsight} has proposed methods that use causal inference to explain query results.

\paragraph*{Multi dimensional data aggregation}
Previous work on multidimensional data aggregation developed methods that extend the traditional drill-down and roll-up operators to find the most interesting data parts for exploration~\cite{agarwal1999cube,sathe2001intelligent,joglekar2017interactive,DBLP:journals/pvldb/YoungmannAP22,10.1145/3448016.3457259}.
Other works have focused on assessing the similarity between data cubes~\cite{baikousi2011similarity}), or discovering intriguing data visualizations~\cite{vartak2015seedb,zhang2021viewseeker}. 

Part of our goal is to identify subpopulations for which we can generate recommendations. We utilize existing solutions whenever applicable (e.g., we use the Apriori algorithm~\cite{agrawal1994fast} to find sufficiently large subpopulations) and develop novel methods when necessary.











\section{Background on Causal Inference}
\label{sec:background-causal}

 \newtextold{In this section, we 
 review the basic concepts and key assumptions for inferring the effects of an intervention on the outcome on collected datasets without performing randomized controlled experiments. 
We use {\em Pearl's graphical causal model} for {\em observational causal analysis} \cite{pearl2009causal} to define these concepts.}

\par
\paratitle{Causal Inference and Causal DAGs} The primary goal of causal inference is to model causal dependencies between attributes and evaluate how changing one variable (referred to as intervention) would affect the other.
Pearl's Probabilistic Graphical Causal Model \cite{pearl2009causal} can be written as a tuple $(\exo, \edvar, Pr_{\exo}, \psi)$, where $\exo$ is a set of {\em exogenous} variables, $\Pr_{\exo}$ is the joint distribution of \exo, and $\edvar$ is a set of observed {\em endogenous variables}.
Here $\psi$ is a set of structural equations that encode dependencies among variables. The equation for $A \in \edvar$ takes the following form:
$$\psi_{A}: 
\dom(Pa_{\exo}(A)) {\times} \dom(Pa_{\edvar}(A)) \to \dom(A)$$
Here $Pa_{\exo}(A) {\subseteq} {\exo}$ and $Pa_{\edvar}(A) {\subseteq} \edvar \setminus \{A\}$ respectively denote the exogenous and endogenous parents of $A$. A causal relational model is associated with a directed acyclic graph ({\em causal DAG}) $G$, whose nodes are the endogenous variables $\edvar$ and there is a directed edge from $X$ to $O$ if  $X {\in} Pa_{\edvar}(O)$. The causal DAG obfuscates exogenous variables as they are unobserved. 
The probability distribution $\Pr_{\exo}$ on exogenous variables $\exo$ induces a probability distribution  
on the endogenous variables $\edvar$ by the structural equations $\psi$.  A causal DAG can be constructed by a domain expert as in the above example, or using existing {\em causal discovery} algorithms~\cite{glymour2019review}.

\begin{figure}
    \centering
    \small
    \begin{tikzpicture}[node distance=0.6cm and 1cm, every node/.style={minimum size=0.5cm}]
        \tikzset{vertex/.style = {draw, circle, align=center}}

        \node[vertex] (Ethnicity) {\bf\scriptsize{{Ethnicity}}};
        \node[vertex, right=0.3cm of Ethnicity] (Gender) {\bf{\scriptsize{Gender}}};
        \node[vertex, right=0.3cm of Gender] (Age) {\bf{\scriptsize{Age}}};
        \node[vertex, below=0.3cm of Gender] (Role) {\bf{\scriptsize{Role}}};
        \node[vertex, right=0.3cm of Role] (Education) {\bf{\small{\scriptsize{Education}}}};
        \node[vertex, below=0.3cm of Role] (Salary) {\bf{\scriptsize{Salary}}};

        \draw[->] (Ethnicity) -- (Salary);
        \draw[->] (Gender) -- (Role);
        \draw[->] (Age) -- (Role);
         \draw[->] (Education) -- (Role);
           \draw[->] (Education) -- (Salary);
             \draw[->] (Ethnicity) -- (Education);
                \draw[->] (Ethnicity) -- (Role);
             \draw[->] (Gender) -- (Education);
               \draw[->] (Age) -- (Education);
                 \draw[->] (Role) -- (Salary);
        \draw[->] (Gender) to[bend right] (Salary);
        \draw[->] (Age) -- (Salary);
    \end{tikzpicture}
    \caption{Partial causal DAG for the Stack Overflow dataset.}
    \label{fig:causal_DAG}
\end{figure}
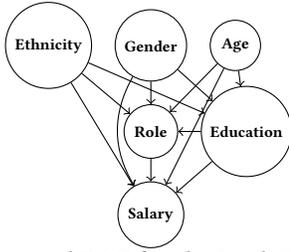

 \begin{example}
Figure \ref{fig:causal_DAG} depicts a partial causal DAG for the SO dataset over the attributes in Table \ref{tab:data} as endogenous variables (we use a larger causal DAG with all 20 attributes in our experiments). 
  Given this causal DAG, we can observe that the role that a coder has in their company depends on their education, age gender and ethnicity.
\end{example}
\par

\par
\paratitle{Intervention} In Pearl's model, a treatment $T = t$ (on one or more variables) is considered as an {\em intervention} to a causal DAG by mechanically changing the DAG such that the values of node(s) of $T$ in $G$ are set to the value(s) in $t$, which is denoted by $\doop(T = t)$. Following this operation, the probability distribution of the nodes in the graph changes as the treatment nodes no longer depend on the values of their parents. Pearl's model gives an approach to estimate the new probability distribution by identifying the confounding factors $Z$ described earlier using conditions such as {\em d-separation} and {\em backdoor criteria} \cite{pearl2009causal}, which we do not discuss in this paper.

\par
\paratitle{Average Treatment Effect} The effects of an intervention are often measured by evaluating
{\em Conditional Average treatment effect (CATE)}, 
measuring the effect of an intervention on a subset of records~\cite{rubin1971use,holland1986statistics} by calculating the difference in average outcomes between the group that receives the treatment and the group that does not (called the {\em control} group), providing an estimate of how the intervention by $T$ influences an outcome $O$ for a given subpopulation. 
Given a subset of the records defined by (a vector of) attributes $B$ and their values $b$, 
we can compute $CATE(T,O \mid B = b)$ as:
{
\begin{eqnarray}    
    \mathbb{E}[O \mid \doop(T=1), B = b] -  
    \mathbb{E}[O \mid \doop(T=0), B = b]\label{eq:cate}
\end{eqnarray}
}
Setting $B=\phi$ is equivalent to the ATE estimate.
The above definitions assumes that the treatment assigned to one unit does not affect the outcome of another unit (called the {Stable Unit Treatment Value Assumption (SUTVA)) \cite{rubin2005causal}}\footnote{This assumption does not hold for causal inference on multiple tables and even on a single table where tuples depend on each other.}.

The ideal way of estimating the ATE and CATE is through {\em randomized controlled experiments}, 
where the population is randomly divided into two groups (treated and control, for binary treatments): 
denoted by 
$\doop(T = 1)$ 
and $\doop(T = 0)$ resp.)~\cite{pearl2009causal}.
However, randomized experiments cannot always be performed due to ethical or feasibility issues. In these scenarios, observational data is used to estimate the treatment effect, which requires the following additional assumptions. 
\newtextold{
The first assumption is called {\em unconfoundedness} or {\em strong ignorability}  \cite{rosenbaum1983central} says that the independence of outcome $O$ and treatment $T$ conditioning on a set of confounder variables  (covariates) $Z$, i.e.,
 $    O \independent T | Z {=} z$.
The second assumption called {\em overlap or positivity} says that there is a chance of observing individuals in both the treatment and control groups for every combination of covariate values, i.e., 
   $ 0 < Pr(T {=} 1 ~~|~~Z {=} z)< 1 $.
}
The unconfoundedness assumption requires that the treatment $T$ and the outcome $O$ be independent when conditioned on a set of variables $Z$. In SO, assuming that only $Z$ =\{\verb|Gender|, \verb|Age|, \verb|Country|\} affects $T = $ \verb|Education|, if we condition on a fixed set of values of $Z$, i.e., consider people of a given gender, from a given country, and at a given age, then $T = $ \verb|Education| and $O = $ \verb|Salary| are independent. For such confounding factors $Z$,  Eq. (\ref{eq:cate}) reduces to the following form 
(positivity 
gives the feasibility of the expectation difference): 
 \vspace{-1mm}
{\small
\begin{flalign}    
 & CATE(T,O {\mid} B {=} b) {=} \nonumber
    \mathbb{E}_Z \left[\mathbb{E}[O {\mid} T{=}1, B {=} b, Z {=} z] {-}  
    \mathbb{E}[O {\mid} T{=}0, B {=} b, Z {=} z]\right]\label{eq:conf-cate}
\end{flalign}
}
This equation contains conditional probabilities and not $\doop(T = b)$, which can be estimated from an observed data. 
Pearl's model gives a systematic way to find such a $Z$ when a causal DAG is available.

\section{Problem Formulation}
\label{sec:problem}

We consider a single-relation database over a schema $\attrset$. The schema is a vector of attribute names, i.e., $\attrset {=} (A_1, \ldots, A_s)$, where each $A_i$ is associated with a domain $\dom(A_i)$, which can be categorical or continuous. 
A database instance \db, populates the schema with a set of tuples $t {=} (a_1, \ldots, a_s)$ where $a_i {\in} \dom(A_i)$. We use $t[A_i]$ to denote the value of attribute $A_i$ of tuple $t$.  

\newtextold{Our high-level goal is to return a set of {\em prescription rules} (ruleset) with certain desired properties including fairness. In this section, first we define patterns on attributes, protected groups, prescription rules, and discuss the desired properties, finally defining our optimization problem in \Cref{subsec:problem}.}

\subsection{Pattern and Protected Group}
\label{subsec:rules}
To define the notion of \emph{prescription rules}, we build upon the commonly used concept of \emph{patterns}~\cite{roy2014formal,wu2013scorpion,el2014interpretable,lin2021detecting}.
\revc{Patterns are commonly used  in query result explanation research~\cite{roy2014formal,wu2013scorpion,el2014interpretable,lin2021detecting}. Generally, patterns are equivalent to the WHERE clause in SQL queries. However, in this work, we focus on a specific type of pattern: a conjunction of predicates, as has been done in prior query result explanation research~\cite{roy2014formal,wu2013scorpion,el2014interpretable,lin2021detecting,DBLP:journals/pacmmod/YoungmannCGR24}.}


\begin{definition}[Pattern]\label{def:pattern}
Given a database instance \db\ with schema \attrset,
a predicate is an expression of the form $\varphi {=} A_i ~{\tt op }~ a_i$, 
where $A_i {\in} \attrset$, 
$a_i {\in} \dom(A_i)$, and ${\tt op} {\in} \{=,\neq, <, >, \leq, \geq\}$.  
A {\em pattern} is a conjunction of predicates $\pattern {=} \varphi_1 \land \ldots \land \varphi_k$. 
\end{definition}

\begin{example}
  An example pattern over the Stack Overflow dataset (\cref{tab:data}) is $\pattern =$\verb|{Role = Designer| $\land$ \verb|Country = US}|. It defines a subset of the dataset comprised of designers from the US. 
\end{example}
 
Next we define {\em coverage} of a pattern $\pattern$ defined by the number of tuples from $\db$ that it captures.
\begin{definition}[Coverage of a pattern]\label{def:coverage}
Given a database instance \db\, a pattern $\pattern$, and a tuple $t {\in} \db$, $\pattern$ is said to {\em cover} $t$ if $t$ satisfies the predicates in $\pattern$. 
The subset of tuples in $\db$ covered by $\pattern$ 
is denoted by 
$\coverage(\pattern)$.
\end{definition}

As is common in fairness research~\cite{stoyanovich2020responsible,zehlike2017fa,caton2024fairness}, we assume the presence of a protected group, defined by the pattern $\pattern_p$. The remainder of the population (i.e., $\db {\setminus} \pattern_p(\db)$) is referred to as the non-protected group.
A protected group in the Stack Overflow dataset may be defined based on sensitive attributes such as age or ethnicity. For instance, it could be defined as $\pattern_p = \{$ \verb|Ethnicity|$\neq$ \verb|White|$\}$ to refer to non-white individuals.

\subsection{Prescription Rules}
A \emph{prescription rule} outlines an \emph{intervention} \newtextold{({\em treatment})} designed to improve a target variable within a particular subpopulation. For example, a prescription rule might recommend that people under 25 pursue a Ph.D. to increase their salary. 
Before defining perception rules, we first discuss \newtextold{how attributes participate in such rules}.

\paragraph*{Mutable and immutable attributes}
We assume the attributes $\attrset$ are partitioned into two disjoint sets. The first set contains the interventional attributes (e.g., programming language, education), which are the attributes that can be changed to improve the outcome. The second set contains immutable attributes (e.g., age, gender), which cannot be changed \newtextold{through prescription}. 
Formally, let $\immutable \subseteq \attrset$, denote the set of immutable attributes and $\mutable \subseteq \attrset$ denote the set of mutable (interventional) attributes, where $\mutable \cap \immutable = \emptyset$ and the outcome $O \notin \mutable \cup \immutable$. 
\newtextold{This categorization 
intends to prohibit the infeasible or impractical recommendations to increase the outcome (e.g., changing one's age or ethnicity to improve one's income). }

\newtextold{Our prescription rules defined below as a combination of grouping and intervention patterns are motivated by the causal explanations defined in \cite{DBLP:journals/pacmmod/YoungmannCGR24}. However, the focus of this paper is to study the interplay between utility and fairness of a ruleset that was not considered in \cite{DBLP:journals/pacmmod/YoungmannCGR24}. As a result, the specific objectives and optimization problem are defined differently as discussed next.

\begin{definition}[Prescription Rule and Ruleset, Grouping and intervention patterns, and Coverage]
Given a database $D$ with mutable attributes $\mutable$ and immutable attribute $\immutable$, 
a {\em prescription rule} $r$ is a pair of patterns $r = (\patterngroup, \patterninterv)$, where (1) $\patterngroup$ is called the  {\em grouping pattern} and consists exclusively of the immutable attributes in $\immutable$, and (2)  $\patterninterv$ is called the {\em intervention pattern} and consists exclusively of the mutable attributes in $\mutable$.

By overloading notations, $\coverage(r) = \coverage(\patterngroup)$ is called the {\em coverage of rule $r$} since it captures the subset of tuples in $\db$ on which the rule $r$ applies, i.e., each tuple $t\in \db$ is either \emph{covered} or not by $\patterngroup$ of rule $r$. 
$\patterninterv$ defines the recommended intervention in the prescription rule aimed at betterment of the outcome $O$ for the subgroup that $\coverage(r)$ defines.

Given a set of prescription rules $R$ (called a ruleset), $\coverage(R) = \cup_{r \in R} \coverage(r)$, i.e., coverage of a ruleset corresponds to the subset of tuples in $\db$ that are covered by at least one of the rules in $R$.


\end{definition}


For a prescription rule $r = (\patterngroup, \patterninterv)$,  the intervention pattern $\patterninterv$ partitions the tuples defined by $\patterngroup$ into treated ($T = 1$ if $\patterninterv$ evaluates to true for a tuple) and control groups ($T = 0$ if $\patterninterv$ evaluates to false). This partition is then used to assess the causal effects of the intervention $\patterninterv$ on the outcome $O$ within the subpopulation $\coverage(r)$ which the rule $r$ applies to.  
}

\newtextold{
\begin{example}
An example prescription rule suggests that individuals aged 25-34 with dependents, should work as front-end developers. 
($\patterngroup: \texttt{age = 25-34} \land \texttt{dependents = yes}$), the intervention is working as front-end developers ($\patterninterv: \texttt{role = {frontend developer}}$). The expected CATE value is \$44,009, namely, the expected salary increase for a 25-34-year-old individual with dependents working as a frontend developer is \$44,009 per year (compared to a 25-34-year-old individual with dependents working in a different role).
\end{example}
}



\subsection{Utility of a Prescription Ruleset}
\label{subsec:utility}
\newtextold{{\bf Utility of a single rule:} To evaluate the effectiveness of a prescription rule $r = (\patterngroup, \patterninterv)$ toward improving the outcome $O$, we define its utility. The utility measures the expected impact (as CATE) of the recommended intervention on the outcome $O$ within the subpopulation $\coverage(r)$ where $r$ applies to. 
We define the overall utility, and utility for the protected and non-protected groups.


\begin{definition}[Utility of a prescription rule - overall, protected, non-protected]\label{def:utility}
Given a database instance \db\ with schema \attrset, an outcome variable $O$, a causal model \model\ on \attrset, 
a protected group $p$ defined by a pattern $\pattern_p$, and a prescription rule $r =( \patterngroup, \patterninterv)$,
\begin{enumerate}
    \item the {\em overall utility} of $r$ is defined as:
\begin{equation}
    utility(r) {:=} CATE_{\model}(\patterninterv, O~|~\patterngroup)
    \label{eq:utility-overall}
\end{equation}
\item the {\em  utility of $r$ for the protected group} $p$ is defined as:
\begin{equation}
utility_p(r) {:=} CATE_{\model}(\patterninterv, O~|~\patterngroup \land \pattern_p)    \label{eq:utility-protected}
\end{equation}
\item the {\em  utility of $r$ for the non-protected group} $p$ is defined as:
\begin{equation}
utility_{\bar{p}}(r) {:=} CATE_{\model}(\patterninterv, O~|~\patterngroup \land \pattern_{\neg p})\label{eq:utility-non-protected}
\end{equation}
\end{enumerate}

The subscript $\model$ denotes that the CATE is estimated using the causal model, \newtextold{and we drop the subscript when it is clear from context.}

If $\coverage(r) = \coverage(\patterngroup) = \emptyset$, i.e., if the rule does not apply to any tuple in $D$, then we assume that $utility(r) = 0$; similarly $utility_p(r) = 0$ if $\coverage(\patterngroup \land \pattern_p) = \emptyset$, and $utility_{\bar{p}}(r) = 0$ if $\coverage(\patterngroup \land \pattern_{\neg {p}}) = \emptyset$.
\end{definition}
}


\newtextold{The goal of prescription rules is to improve the outcome $O$ as desired. }
If the goal is to increase the outcome $O$ (e.g., increase salary), we discard rules with negative utility, as they do not help achieve this objective. Similarly, if the aim is to decrease the outcome, we ignore rules with negative utility. Throughout the paper, without loss of generality, we assume that the goal is to increase the outcome, thereby focusing on maximizing utility.


\cut{
Furthermore, it is necessary to ensure that the positivity assumption is met (see \cref{sec:background-causal}). 
\newtextold{Note that the grouping pattern $\patterngroup$ can only comprise immutable attributes and intervention patterns $\patterninterv$ can only comprise mutable attributes.}
\sr{commented out the next paragraph - but then what.}
}

\cut{
Given a prescription rule $r = (\patterngroup, \patterninterv)$, \red{if the attributes used
to define the grouping pattern $\patterngroup$ are also used to define the intervention pattern $\patterninterv$, either the control or treatment group might be empty}. \sr{this cannot be possible -- one from mutable other immutable -- this paragraph should go} For example, assume $\patterngroup = $ \verb|{County = US}|, which defines the subpopulation of people residing in the US. Setting the intervention pattern $\patterninterv$ to include the attribute \verb|Country| will may lead the intervention being empty (e.g., if $\patterninterv = $\verb|{Country = China}|, as the set of people leaving in China among people from the US is empty). If on the other hand, $\patterninterv$ contains, among other predicates, the predicate \verb|{County = US}|, then this predicate can be removed from $\patterninterv$ without affecting its definition of control and treatment groups (as both the treatment and control group includes only people from the US). 
}


\newtextold{
\textbf{Prescription to individuals when multiple rules apply:}} 
\newtextold{When dealing with a ruleset $R$, it is possible for multiple rules to apply to the same subpopulation. Specifically, if two rules $r_i {=} (\patterngroup^i, \patterninterv^i)$ and $r_j {=} (\patterngroup^j,\patterninterv^j) {\in} R$ share a non-empty intersection between their coverage, namely $\coverage(\patterngroup^i) {\cap} \coverage(\patterngroup^j) {\neq} \emptyset$, then the subpopulation defined by the pattern $\patterngroup^i {\wedge} \patterngroup^j$ will have more than one rule. 
}
In our definition below for utility of a ruleset, we refrain from applying more than one rule to a subpopulation for two reasons. First, two rules may conflict with each other. For instance, if one rule suggests individuals above \newtextold{25 to earn a Ph.D.}, while another \newtextold{recommends women over 20 pursue an MBA}, 
women above 25 would receive conflicting recommendations. 
Second, CATE is known to be non-monotonic~\cite{DBLP:journals/pacmmod/YoungmannCGR24}, implying that appending a predicate to an intervention pattern can either increase or decrease the CATE value.
\newtextold{
Therefore, employing multiple rules simultaneously for a subpopulation might yield a utility gain smaller than the individual rules. Hence when multiple rules apply to a tuple in $\db$, we assume that only one is chosen by the decision-maker.
}

\noindent
\smallskip
\newtextold{{\bf Utility of a ruleset:} For a prescription ruleset $R$, we use its {\em expected utility} on $\db$ as the utility of $R$.}

\newtextold{
\begin{definition}[Expected Utility of a Ruleset]\label{def:expected-utility}
    The {\em expected utility} of a prescription ruleset $R$ is defined as the 
    average maximum utility of an individual from $\coverage(R)$ from the rules in $R$ that applies to the individual, i.e.,  
\begin{equation}
    \text{ExpUtility}(R) = \frac{1}{n} \sum_{t \in \coverage(R)} \max_{r \in R_t} (\text{utility}(r))
\label{eq:exp-utility-all}
\end{equation}

where $R_t \subseteq R$ denotes the set of rules covering the tuple $t$, and $n = |\db|$. Note that if a rule does not apply to a tuple, its utility is zero, so the sum above is also over all $t \in D$. 

Given a protected pattern $\pattern_p$, the expected utility for the protected and non-protected groups are defined as follows:
\begin{eqnarray}
    \text{ExpUtility}_p(R) & = &  \frac{1}{n_p} \sum_{t \in \coverage_p(R)} \min_{r \in R_t} (\text{utility}(r)) \label{eq:exp-utility-protected}\\
    \text{ExpUtility}_{\bar{p}}(R) & = & \frac{1}{n_{\bar{p}}} \sum_{t \in \coverage_{\bar{p}}(R)} \max_{r \in R_t} (\text{utility}(r)) \label{eq:exp-utility-non-protected}
\end{eqnarray}
where $\coverage_p(R)$ denotes the set of protected individuals covered by $R$ and 
$n_p = |\coverage_p(R)|$
(similarly $\coverage_{\bar{p}}()$ and $n_{\bar{p}}$).
\end{definition}
}

\cut{
Similarly, we define the expected utility for the protected group by calculating the expected utility for a randomly sampled protected individual:
\[
\text{ExpUtility}_p(R) = \frac{1}{n_p} \sum_{t \in coverage_p(R)} \min_{r \in R_t} (\text{utility}(r))
\]
where $coverage_p()$ denotes the set of protected individuals covered by a set of rules $R$ and $n_p = |\pattern_p(\db)|$.

In a similar manner, we define the expected utility for a randomly sampled non-protected individual, referred to as the expected utility of non-protected.
}

\cut{Our goal is to generate a prescription ruleset $R$ where the outcome is fixed across all rules (e.g., all rules aim to increase salary). These rules may overlap, but they all target the same outcome variable. }

\newtextold{
Note the difference between formulas (\ref{eq:exp-utility-protected}) and (\ref{eq:exp-utility-all}, \ref{eq:exp-utility-non-protected}).
Since we do not assume any restriction on which rule is chosen for a tuple when multiple rules apply, we do a conservative worst-case analysis on fairness.
We assume that protected individuals choose the worst possible rule, while the rest choose the best possible one. This ensures that the expected utility for the protected group in reality (irrespective of the rule chosen for each protected tuple) will be at least as high as the expected utility from the least beneficial relevant prescription rule for this group.
}


\subsection{Size of a Prescription Ruleset}\label{subsec:size}

\cut{
Given an outcome variable $O \in \attrset$ (e.g., salary), our goal is to find a set of prescription rules to increase $O$ for all individuals and ensure that the protected group will ultimately be in a better condition than before applying the rules (a fairness constraint). We favor sets with a small number of effective rules, that apply to a large portion of the population (coverage constraint). We begin by discussing the size and utility objectives. We then introduce the coverage and fairness constraints. Finally, we present our problem definition.


Given a set $R$ of prescription rules, we define its size and expected utility as follows. 

\noindent
\textbf{Size}:
}
\newtextold{The size of a prescription ruleset is the number of rules in that set, denoted by $size(R)$. Ideally, we want to find a small-size ruleset. The intuition is that, }
the fewer the rules in a set, the easier it is to understand the suggested interventions. 
\newtextold{Suppose we want to find a ruleset with high utility without specifying a constraint on the size. The following lemma shows that the best strategy is to return the {\em optimal rule} that applies to each individual.}
\newtextold{That is, to maximize utility, prescribing a personalized rule for each individual may lead to the best utility. Specifically, }
we can show that for every rule $r = (\patterngroup, \patterninterv)$, there exists a subgroup $g' {\subseteq} \coverage(\patterngroup)$ and a intervention $\pattern_{t'}$ s.t the utility of  $r' {=} (g', \pattern_{t'})$ is greater than that of the original rule $r$. 


\begin{lemma}
\label{lemma:individual_rules}
   Given a rule $r {=} (\patterngroup, \patterninterv)$, there exists a rule $r' {=} (\pattern_{g'}, \pattern_{t'})$ s.t $\pattern_{g'} {\subset} \coverage(\patterngroup)$ and $utility(r') {\geq} utility(r)$. 
\end{lemma}

\newtextold{This property implies that the number of prescription rules in the optimal solution is $O(|D|)$, making it impractical to implement in real-world scenarios. For instance, consider a policy enacted by a government official to allocate healthcare resources based on patient data. If the number of rules scales linearly with the size of the dataset, it would become infeasible to apply the policy effectively across a large population. }
Therefore, we limit the number of recommended rules.
One approach is to impose a strict limit on the number of rules selected. However, pre-setting this constraint often requires tuning to balance utility and comprehensibility. Therefore, we incorporate the number of rules as an objective, considering rulesets with fewer rules to be more desirable, as was done in \cite{lakkaraju2016interpretable}.



\cut{
\vspace{1mm}
\noindent
\textbf{Rule selection}:
As outlined in \cref{subsec:utility}, when multiple rules apply to an individual, we assume they choose only one. For worst-case analysis, we assume that non-protected individuals select the worst possible rule, while the rest choose the best possible one (using min for protected expected utility calculation and max for expected utility and non-protected expected utility calculation). 

\sr{read up to here}
}

\subsection{Coverage Constraints}
\label{subsec:coverage_constraint}
We consider two types of coverage constraints: \emph{group coverage}, where the goal is to find a solution that covers a predefined fraction of protected individuals and a certain fraction of the entire population, and \emph{rule coverage}, where every selected rule must cover a certain fraction of the population and protected individuals.\\
\noindent
{\bf{Group Coverage}} Given two thresholds $\theta, \theta_p {\in} [0,1]$, we say that a ruleset $R$ satisfies the group coverage constraint if $R$ covers at least a $\theta$ fraction of the population, and a $\theta_p$ fraction of the protected subpopulation. Formally, both conditions are satisfied: (i) $Coverage(R) {\geq} \theta {\cdot} |\db|$, (ii) $Coverage_p(R) {\geq} \theta_p {\cdot} |\pattern_p(\db)|$, 
where $Coverage_p(R)$ denotes the number of covered protected individuals by $R$. \\
\noindent
{\bf Rule Coverage}
 Given two thresholds $\theta, \theta_p {\in} [0,1]$, we say that a ruleset $R$ satisfies the rule coverage constraint if every rule $r {\in} R$ covers at least a $\theta$ fraction of the population, and at least a $\theta_p$ fraction of the protected subpopulation. Formally, both of the following conditions hold: (i) For every $r {\in} R$: $coverage(r) {\geq} \theta {\cdot} |\db|$, (ii) For every $r {\in} R$: $coverage_p(r) {\geq} \theta_p {\cdot} |\pattern_p(\db)|$,
where $coverage_p(r)$ denotes the number of covered protected individuals by $r$.

\subsection{Fairness Constraints}
\label{subsec:fairness_constraint}
We study two definitions of fairness: statistical parity (SP)~\cite{mehrabi2021survey}, and bounded group loss (BGL)~\cite{agarwal2019fair}. Those definitions are based on equivalent notions of fairness for regression tasks~\cite{agarwal2019fair}. We next provide an extension for these definitions to causal estimates.

Group and individual fairness are two key concepts in algorithmic fairness~\cite{stoyanovich2020responsible,binns2020apparent,garcia2021maxmin}. Group fairness aims to ensure that different groups receive similar outcomes. Individual fairness focuses on treating similar individuals similarly, meaning that if two individuals are alike in relevant aspects, they should receive similar outcomes. Both approaches aim to reduce bias, with the choice of which approach to adopt depending on the specific context.
Next, we present four types of fairness constraints: SP and BGL, each of which can be applied to ensure group or individual fairness.

\subsubsection{Statistical parity}
In SP, the goal is to ensure that the gain in the utility of a protected individual is similar to that of any individual from the non-protected group.

\noindent
\textbf{Group Fairness}:
Intuitively, if we randomly sample a protected individual, the expected gain should be almost the same as that of an individual from the non-protected group.
Formally:\\
$|\text{ExpUtility}_p(R) {-}\text{ExpUtility}_{\bar{p}}(R)| {\leq} \epsilon$,
 where $\epsilon {>} 0$ is a threshold. 

\vspace{1mm}
\noindent
\textbf{Individual Fairness}:
Individual fairness says that the expected gain of every protected individual is similar to that of an individual from the non-protected group. That means that the expected utility of each rule $r{\in} R$ on a protected individual should be similar to that of an individual from the non-protected group.
Formally, for every $r{\in} R$,
$|utility_p(r) {-} utility_{\bar{p}}| {\leq} \epsilon$, where $\epsilon {>} 0$ is a threshold.

\subsubsection{Bounded group loss (BGL)}: Fair regression with BGL minimizes the overall loss while controlling the worst loss in the protected group~\cite{agarwal2019fair}.
In our context, this translates to the following constraint: When selecting an individual from the protected group, the utility increase should exceed a specified threshold $\tau {\geq} 0$.

\noindent
\textbf{Group Fairness}: We aim to ensure that the expected utility of a randomly sampled protected individual within $Coverage(R)$ is above a given threshold $\tau$.
Formally,
$\text{ExpUtility}_p(R) {\geq} \tau$.

\smallskip
\noindent
\textbf{Individual Fairness}: We aim to ensure that the gain of every protected individual from $Coverage(R)$ exceeds a threshold $\tau$.
Therefore, a ruleset $R$ satisfies the individual loss constraint if the utility of every rule $r$ on protected individuals is at least $\tau$. Formally, for every rule $r{\in} R$, $utility_p(r) {\geq} \tau$.



\subsection{The \probName\ Problem}\label{subsec:problem}
We are finally ready to present the problem we study in this paper. If we did not have \newtextold{{\em fairness or coverage constraints}}, then our goal is to select a small-size perception ruleset \newtextold{with high expected utility.} 
\newtextold{However, as demonstrated in the introduction, 
not considering fairness constraints may result in a ruleset that are only highly beneficial to 
a small, non-protected subset of the population. Therefore, we extend our problem definition to include coverage and fairness constraints.  
We can apply any of SP or BGL group or individual fairness constraints  (\Cref{subsec:fairness_constraint}), as well as rule or group coverage constraints \Cref{subsec:coverage_constraint}), along with no fairness or coverage constraints, resulting in 18 distinct problem variants. The choice of which constraints to apply is left to the user as it may be application-dependent, and is discussed below. To define the generic problem, we use $R {\models} \fairnessconstraints$ and $R {\models} \coverageconstraints$ to denote that a ruleset $R$ satisfies a given fairness constraint $\fairnessconstraints$ and a given coverage constraint $\coverageconstraints$ respectively (if no constraints are given, these conditions are trivially satisfied). 
We assume that a set of candidate rules $\{r_i\}_{i=1}^l$ has been already mined and is available as an input to the problem.}

\newtextold{
\begin{definition}[\probName\ under fairness and coverage constraints]
\label{def:problem}
Given a database $\db$, a causal model \model, an outcome attribute $O$, a fairness constraint $\fairnessconstraints$, a coverage constraint $\coverageconstraints$, and a collection of prescription rules $\{r_i\}_{i=1}^l$, a subset $R {\subseteq} \{r_i\}_{i=1}^l$ of prescription rules is called {\em valid} if (1) $R \models \fairnessconstraints$ and (2) $R \models \coverageconstraints$.
The goal is to find a valid subset of rules $R^* {\subseteq} \{r_i\}_{i=1}^l$ s.t 
\begin{equation}
R^* {=} argmax_{R {\subseteq} \{r_i\}_{i=1}^l} \left[ \lambda_1 {\cdot} (l - size(R)) {+} \lambda_2 {\cdot} ExpUtility(R) \right]
\label{eq:optimization}
\end{equation}
where $\lambda_1,\lambda_2$ are non-negative weights. 
\end{definition}
}

\begin{figure*}[t]
    \centering
    
\includegraphics[scale=0.43]{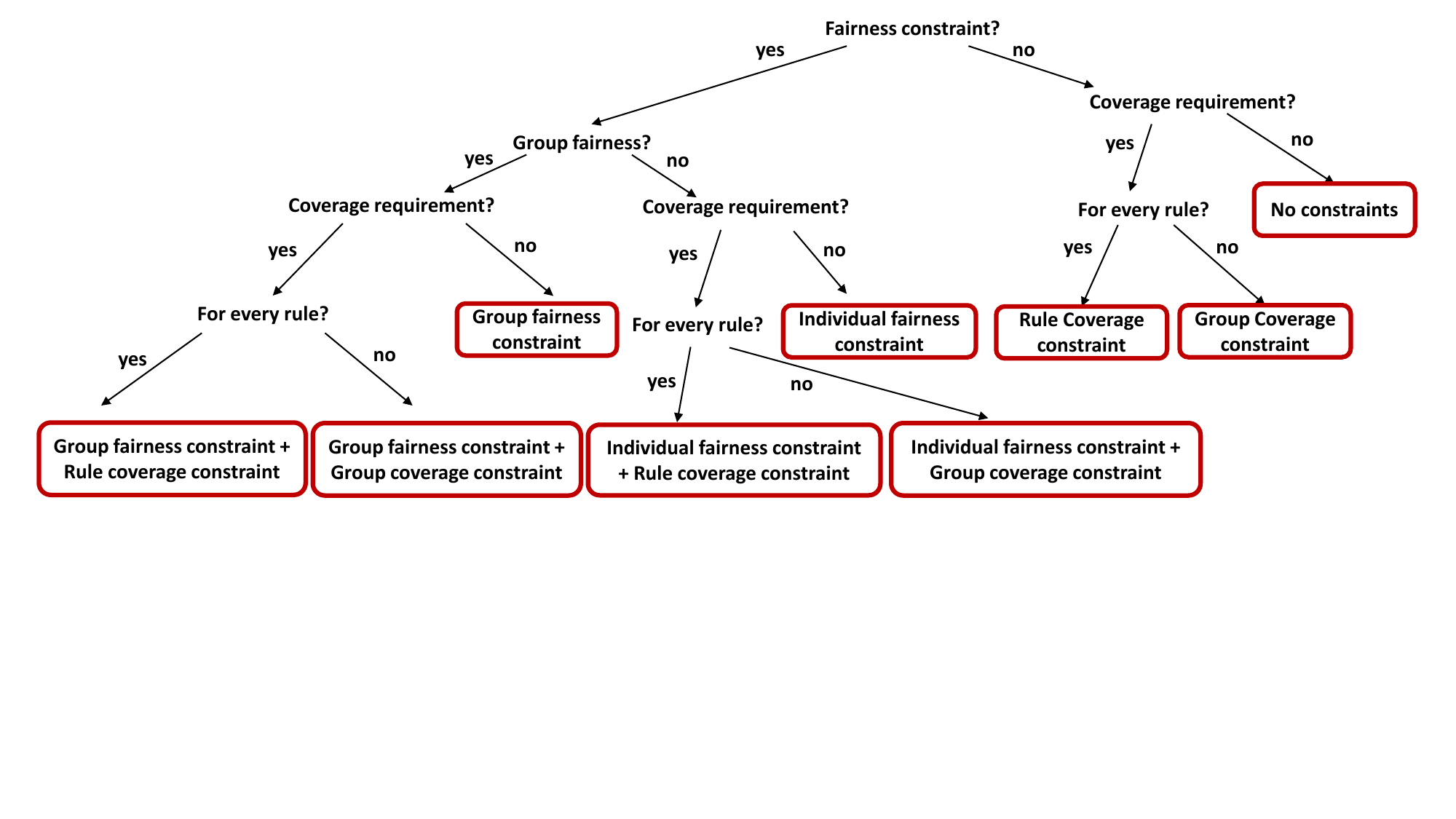}
\vspace{-32mm}
    \caption{A decision tree for selecting the appropriate problem variant.}
    \label{fig:decision_tree}
    \vspace{-5mm}
\end{figure*}

$\lambda_1$ and $\lambda_2$ may be tuned by the user. \newtextold{The above optimization problem, as expected, is NP-hard even for simple variants, although some constraints are matroid constraints and therefore are amenable to greedy approaches (discussions and proofs in the Appendix)}. 
%
 Therefore, we obtain efficient algorithms that work well in practice in \Cref{sec:algo} and experimentally demonstrate the effect of different constraints on the results in \Cref{sec:casestudy}. 


Since we have several options for fairness and coverage constraints, a natural question is which version to use. 
We observe that there is no one-size-fits-all solution 
and the best choice depends on the specific application. For instance, going for individual fairness gives a stronger fairness guarantee at the expense of possible lower utility to everyone. In addition, the complexity of different versions can vary. To assist in making this decision, we summarize the process through a decision tree that guides users in selecting the most suitable variant for their needs, presented in Figure \ref{fig:decision_tree}. The decision to choose between SP or BGL fairness is left to the user.
In Section \ref{sec:casestudy}, we present a case study that empirically compares the obtained rulesets under different problem variants.

\section{The \sysName\ Algorithm}
\label{sec:algo}

A brute-force approach, which considers all grouping and intervention patterns to form prescription rules, results in long runtimes (as we show in Section \ref{exp:scalability}). Instead, we propose a more efficient algorithm, called \sysName\ (\underline{Fair} \underline{CA}usal \underline{P}rescription), which avoids generating every possible prescription rule. \sysName\ can be adapted for any variant of the \probName\ problem. For simplicity, we first describe \sysName\ for the case with SP group fairness and group coverage constraints. We then explain how it can be modified to accommodate other variants.

Our algorithmic framework is outlined in ~\cref{algo:full_algo}.
\sysName\ consists of three parts: (1) generating grouping patterns by using the Apriori algorithm~\cite{agrawal1994fast} (line \ref{line:init}), (2) identifying promising intervention patterns for each grouping pattern by using a lattice traversal approach \cite{asudeh2019assessing}, and (3) finding a set of prescription rules using a greedy approach. We leverage existing solutions (e.g., \cite{asudeh2019assessing, agrawal1994fast,pastor2021looking,DBLP:journals/pacmmod/YoungmannCGR24}) where applicable, and develop novel techniques where necessary.
Specifically, the first step follows the same approach as CauSumX ~\cite{DBLP:journals/pacmmod/YoungmannCGR24}, while the second and third steps introduce novel methods.


\reva{We show that the prescription rule selection problem is
NP-hard even in simple settings (proofs deferred to the Appendix), and therefore developing effective heuristics considering several constraints is non-trivial.
\sysName\ avoids generating all possible rules  (as their number grows exponentially with the database size) and therefore does not perform an exhaustive search and may not return an optimal answer.} If steps 1 and 2 were replaced by a brute-force approach that generates all rules, then a greedy approach for selecting a ruleset could approximate the optimal solution for certain problem variants, as the objective is a non-negative, monotone submodular function (even with a rule coverage or individual fairness constraints which are matroid constraints). However, 
other constraints are harder to satisfy. Future work will explore the complexity of various problem variants and establish theoretical bounds for finding approximate solutions. 

\reva{Despite the fact \sysName\ does not provide formal guarantees for the prescription ruleset selection problem, we emphasize that each selected rule represents an intervention that is statistically significant. Specifically, based on causal
analysis, the expected utility reflects the anticipated average increase in the outcome for the specific subpopulation to which the
rule applies.}
\begin{algorithm}[t]
\footnotesize
\SetKwInOut{Input}{input}\SetKwInOut{Output}{output}
\LinesNumbered
\Input{A database relation \db, a protected group defined by the pattern $\pattern_p$ and an outcome variable $O$}
\Output{A set $\Phi$ 
of prescription rules.} \BlankLine
\SetKwFunction{NextBestExplanation}{\textsc{NextBestExplanation}}
\SetKwFunction{SolveLP}{\textsc{SolveLP}}
\SetKwFunction{Greedy}{\textsc{ApplyGreedy}}
\SetKwFunction{GetGroupingPatterns}{\textsc{GetGroupingPatterns}}
\SetKwFunction{ExplanationSummary}{\textsc{ExplanationSummary}}
\SetKwFunction{GetTreatment}{\textsc{GetIntervention}}
$\Phi \gets \emptyset$\;
$\mathcal{G} \gets \GetGroupingPatterns (\db, O)$\tcp*{\cref{subsec:grouping_patterns}} \label{line:init}   
 \For{$\pattern_g \in \mathcal{G}$}{\label{l:iterate-candidates}
    $\pattern_t \gets \GetTreatment(\pattern_g, O, \pattern_p, \db)$\tcp*{\cref{subsec:treatment_patterns}}\label{l:top-treatment}
    $\Phi \gets \Phi \cup (\pattern_g, \pattern_p)$\\
 }
  $\Phi \gets \Greedy(\Phi, O, \pattern_p)$\tcp*{\cref{subsec:step_3}}\label{l:step3}

\Return $\Phi$\label{line:returnTop}\\
\caption{The \sysName\ algorithmic framework.}\label{algo:full_algo}
\end{algorithm}

\subsection{Step 1: Mining Grouping Patterns}
\label{subsec:grouping_patterns}
Considering every possible grouping pattern is infeasible as their number is exponential ($O(agrmax_{A_i {\in} \attrset}|\dom(A_i)|^{|\attrset|}$). 
Instead, as done in previous work~\cite{DBLP:journals/pacmmod/YoungmannCGR24,pastor2021looking}, we utilize the Apriori algorithm~\cite{agrawal1994fast} to generate candidate patterns. The Apriori algorithm gets a threshold $\tau$, and ensures that
the mined grouping patterns are present in at least $\tau$ tuples of $\db$.
The algorithm guarantees that each mined pattern covers at least $\tau$ tuples from $\db$, making them promising candidates for covering many tuples from $\db$.



\subsection{Step 2: Mining Intervention Patterns}
\label{subsec:treatment_patterns}
Our next goal is to identify an intervention pattern $\patterninterv$ for each mined grouping pattern $\patterngroup$ that maximizes utility (i.e., treatments with the highest CATE for $\patterngroup$) while ensuring fairness to the protected group.
Unlike step 1, this step requires novel techniques for finding treatments that are both fair and have high utility.

Since the number of potential intervention patterns for $\patterngroup$ can be large (exponential in $|\attrset|$), we employ a greedy lattice-traversal~\cite{asudeh2019assessing,DeutchG19} approach, inspired by \cite{DBLP:journals/pacmmod/YoungmannCGR24,pastor2021looking}.
This allows us to materialize and assess the CATE only for promising patterns. 

Concretely, the space of all intervention patterns
can be represented as a lattice where nodes correspond to patterns and there is an edge between $\patterninterv^1$ and $\patterninterv^2$
if $\patterninterv^2$
can be obtained
from $\patterninterv^1$ by adding a single predicate. This lattice can be traversed in a top-down fashion. 
Since not all nodes correspond to treatments with a positive CATE, we only materialize nodes if all their parents have a positive CATE. 
We note that this might lead the algorithm to overlook certain relevant intervention patterns.
However, as shown in \cite{DBLP:journals/pacmmod/YoungmannCGR24}, combining patterns that exhibit a positive CATE is highly likely to result in an intervention with a positive CATE as well. 

When a group fairness constraint is imposed, instead of searching for the treatment with the highest CATE, we search for the treatment that is "fair" by that it maximizes CATE for both protected and non-protected groups, while minimizing disparities.

To identify the most fair treatment, we define the \emph{benefit} of an intervention pattern as follows. Intuitively, we penalize the treatment based on the difference between the utility for the non-protected group and the utility provided to the protected group. The larger the difference, the lower the benefit of the treatment. Formally, the benefit of a rule $r = (\patterngroup,\patterninterv)$ is defined as: 
\vspace{-2mm}
\[
    benefit(r)=
\begin{cases}
   \frac{utility(r)}{1+utility_{\bar{p}}(r) - utility_p(r)},& \text{if } utility_{\bar{p}}(r) \geq utility_p(r)\\
   utility(r), &\text{otherwise}
\end{cases}
\]
\vspace{-3mm}


We implement two optimizations to improve efficiency:
\textbf{(i)} we discard attributes that do not have a causal relationship with the outcome, since such attributes have no impact on CATE values. We can detect such attributes by utilizing the input causal DAG.
\textbf{(ii)} The process of extracting intervention patterns for each grouping pattern can be performed in parallel since this procedure is dependent only on the grouping pattern.

\subsection{Step 3: A Greedy Approach}
\label{subsec:step_3}
The final step involves finding a solution from the rules mined in Steps 1 and 2. We propose a greedy algorithm that optimizes the problem’s objectives.
Intuitively, the algorithm operates as follows: at each iteration, it selects the next best rule that maximizes expected utility, benefit (as defined in \cref{subsec:treatment_patterns}), and coverage. Once the coverage constraints are met, the focus shifts to maximizing benefit and utility. The algorithm stops when the additional gain becomes negligible, as the number of rules is not predetermined.

Formally, the next best rule is determined as follows. Let $\{r_j\}_{j=1}^l$ denote the candidate rules and $R_i$ is the ruleset selected in the first $i$ iterations.
The score of a rule $r$ w.r.t $R_i$ is defined as:
\vspace{-3mm}
\begin{multline*}
\small
score(r){=} Coverage(R_i {\cup} \{r\}) + benefit(R_i {\cup} \{r\}) + ExpUtility(R_i {\cup} \{r\})
\end{multline*}
\vspace{-1mm}
The next best rule $r_{i+1}$ to add in case the coverage constraints are not met yet is defined as:
$r_{i+1}^* = argmax_{r_{i+1} \in \{r_j\}_{j=1}^l \setminus R_i} score(r_{i+1})$
In case the coverage constraints are met, ignore the coverage term. 
The algorithm stops at the first iteration $i$ where the score of the selected rule $r_i$ falls below a predefined threshold, indicating that the marginal gain from $r_i$ is negligible.

\subsection{Adjustments to Other Variants}
We explain how \sysName\ can be adjusted to solve other problem variants. 
We set the Apriori's threshold to ensure that each mined grouping pattern covers a sufficient number of individuals when a rule coverage constraint is imposed (step 1). Without coverage constraints, the Apriori threshold can be set to any value. 

Without fairness constraints, in Step 2, the goal is to identify the intervention with the highest CATE value (as was done in \cite{DBLP:journals/pacmmod/YoungmannCGR24}). Whit an individual fairness constraint, each rule must satisfy this constraint, so we only select interventions that are guaranteed to meet the constraint while maximizing CATE (step 2).

In case a group BGL fairness constraint is imposed, we define the benefit of a rule $r = (\patterngroup,\patterninterv)$ as follows. Intuitively, we penalize the treatment based on the difference between the minimum required utility for the protected group and the utility provided to the protected group by this treatment. The larger the difference, the lower the benefit of the treatment. Formally: 
\[
    benefit(r)=
\begin{cases}
   \frac{utility(r)}{1 + \tau - utility_p(r)},& \text{if } \tau\geq utility_p(r)\\
   utility(r), &\text{otherwise}
\end{cases}
\]
where $\tau$ is the threshold for the BGL fairness constraint. This benefit definition is applied in Step 2 of the algorithm to identify fair and effective treatments for the mined grouping patterns.


\vspace{1mm}
\noindent
\textbf{Runtime complexity analysis}:
The maximum number of rules in a database $\db$ with attributes $\attrset$ is bounded by $|\db|^{|\attrset|}$ (considering both grouping and intervention patterns and the active domain of attributes), which is polynomial in terms of \emph{data complexity}, assuming a fixed schema \cite{Vardi82}. The final greedy step is also polynomial in the number of rules considered. Additional operations, such as calculating CATE values, are polynomial in $\db$, leading to worst-case polynomial data complexity. 
As we demonstrate in \cref{exp:scalability}, our algorithm is capable of efficiently handling large datasets.

\vspace{-1mm}
\section{Case Study}
\label{sec:casestudy}

The objective of this case study is to evaluate the impact of various constraints on the solution.
We analyze two datasets, (1) German Credit (German in short) and (2) Stack Overflow (SO in short), each with a corresponding protected group, and assess the rules chosen by \sysName\ under different constraints. We aim to understand how these constraints influence coverage, utility, and disparities (for fairness) between protected and non-protected groups. \revb{We present example chosen rules under different configurations. We chose the rules by randomly picking one from each category (one that favors the protected group, one that favors the non-protected, and another that is more balanced). The full lists of rules are available in \cite{fullversion}.}

\vspace{1mm}
\paratitle{Datasets \revb{\& protected groups}}
We examine two commonly used datasets:
(1) Stack Oveflow (SO)~\cite{stackoverflowreport}, as described in \Cref{example:ex1}. Here, the goal is to increase salary. (2) German Credit~\cite{asuncion2007uci}, which contains details of bank account holders, including demographic and financial information. Here, the goal is to increase the credit score (binary). 
The corresponding 
causal DAG was constructed using  
~\cite{youngmann2023causal}. The datasets' statistics are presented in \cref{tab:datasets}. The protected groups were selected to represent subgroups where the desired outcome was relatively low and sufficiently large to ensure the discovery of statistically significant rules. \revb{The protected group in Stack Overflow is defined as individuals from countries with a low GDP, which constitutes 21.5\% of the data (the GDP attribute is categorical in this dataset). In the German data, the protected group is defined as single females, which constitutes 9.2\% of the data.} 

\begin{table}
	\centering
\small
		\caption{Examined datasets.}
			\label{tab:datasets}
	\begin{tabular} 
 {p{8mm}cccp{37mm}}
		\toprule
	\textbf{Dataset} & \textbf{Tuples}& \textbf{Atts}& \textbf{Mut Atts}&\textbf{Protected Group}
	 \\
		\midrule 

SO&38K&20&10&People from countries with a low GDP (\revb{21.5\% of the data})\\
\hline
German&1000&20&15&Single Females (\revb{9.2\% of the data}) \\
				\bottomrule
	\end{tabular}
\end{table}

\begin{table*}[h!]
\centering
\small
\caption{Comparison of Solutions in Terms of Size, Coverage, Expected Utility and Unfairness. \revb{IDS and FRL were used to either (i) replace step 1 of \sysName\ to find grouping patterns; (ii) replace step 2 of \sysName\ to find intervention patterns}.}
\label{tab:problem_variants}
\begin{tabular}{lccccccc}
\toprule
\textbf{Stack Overflow (SP fairness)} & \textbf{\# rules} & \textbf{coverage} & \textbf{coverage pro} & \textbf{exp utility} & \textbf{exp utility non-pro}&\textbf{exp utility pro} &\textbf{unfairness} \\



\midrule 
No constraints  & 20& 99.91\%& 99.98\%& \textbf{32634.2}& \textbf{32626.98}& \textbf{18432.66}& 14194.32 \\
Group coverage  &20& 99.84\%& 99.88\%& 32597.02& 32595.1& 18340.29& 14254.81 \\
Rule coverage  & 10& \textbf{99.99}\%& \textbf{99.99}\%& 22301.77& 22292.02& 16604.92& \textbf{5687.1}\\
Group fairness  & 8& 97.52\%& 97.81\%& 27870.77& 27998.47& 17998.66& 9999.81 \\
Individual fairness   & 20& \textbf{99.99}\%& \textbf{99.99}\%& 28014.58& 28256.35& 14241.07& 14015.28 \\
Group coverage, Group fairness  & 11& 97.95\%& 98.85\%& 27934.76& 28144.58& 18145.23& 9999.35 \\
Rule coverage, Group fairness  & 12& 99.96\%& 99.89\%& 22284.1& 22279.93& 16594.77& 5685.16\\
Group coverage, Individual fairness  & 20& 99.74\%& 99.88\%& 28057.78& 28284.25& 15128.91& 13155.34\\
Rule coverage, Individual fairness  &13& \textbf{99.99}\%& \textbf{99.99\%}& 18591.41& 18606.68& 12797.15& 5809.53 \\

\hdashline[1pt/3pt] 
\revb{IDS (IF clause as grouping pattern)} &\revb{16}&\revb{100\%} &\revb{100\%} &\revb{29770.43}&\revb{29988.1} &\revb{16440.82}& \revb{13547.28}\\

 \revb{IDS (IF clause as intervention pattern)} &\revb{16}&\revb{100\%} &\revb{100\%} &\revb{27763.89}&\revb{27714.9} &\revb{16888.1}& \revb{10826.8}\\


 \revb{FRL (IF clause as grouping pattern)} &\revb{9}&\revb{99.5\%} &\revb{98.85\%} &\revb{27777.43}&\revb{27782.3} &\revb{18891.22}& \revb{8891.08}\\


   \revb{FRL (IF clause as intervention pattern)} &\revb{9}&\revb{100\%} &\revb{100\%} &\revb{28999.22}&\revb{28997.8} &\revb{16453.8}& \revb{12544}\\
        
\midrule

\textbf{German Credit (BGL fairness)} & \textbf{\# rules} & \textbf{coverage} & \textbf{coverage pro} & \textbf{exp utility} & \textbf{exp utility non-pro}&\textbf{exp utility pro}&\textbf{\common{unfairness}}   \\
\midrule 
No constraints  & 17& \textbf{100.0\%}& \textbf{100.0\%}& \textbf{0.39}& \textbf{0.39}& 0.27& 0.12 \\
Group coverage  &18& \textbf{100.0\%}& \textbf{100.0\%}& \textbf{0.39}& \textbf{0.39}& 0.3& 0.09 \\
Rule coverage  & 6& 96.0\%& \textbf{100.0\%}& 0.31& 0.31& 0.3& 0.01\\
Group fairness  & 18& \textbf{100.0\%}& \textbf{100.0\%}& \textbf{0.39}& \textbf{0.39}& 0.3& 0.09 \\
Individual fairness   & 20& \textbf{100.0\%}& \textbf{100.0\%}& 0.37& 0.37& 0.23& 0.14 \\
Group coverage, Group fairness  & 6& \textbf{100.0\%}& \textbf{100.0\%}& 0.36& 0.37& \textbf{0.31}& 0.06 \\
Rule coverage, Group fairness  & 3& 90.0\%& \textbf{100.0\%}& 0.29& 0.29& \textbf{0.31}& \textbf{-0.02}\\
Group coverage, Individual fairness  & 20& \textbf{100.0\%}& \textbf{100.0\%}& 0.37& 0.37& 0.23& 0.14\\
Rule coverage, Individual fairness  &8& 96.8\%& \textbf{100.0\%}& 0.29& 0.29& 0.23& 0.06 \\

 \hdashline[1pt/3pt] 
 \revb{IDS (IF clause as grouping pattern)} &\revb{12}& \revb{100\%}&\revb{100\%} &\revb{0.35}&\revb{0.35} &\revb{0.3}&\revb{0.05} \\

  \revb{IDS (IF clause as intervention pattern)} &\revb{12}& \revb{100\%}&\revb{100\%} &\revb{0.34}&\revb{0.34} &\revb{0.24}&\revb{0.1} \\


 \revb{FRL (IF clause as grouping pattern)} &\revb{13}&\revb{100\%} &\revb{100\%} &\revb{0.26}&\revb{0.26} &\revb{0.21}&\revb{0.05} \\

 \revb{FRL (IF clause as intervention pattern)} &\revb{13}&\revb{100\%} &\revb{100\%} &\revb{0.3}&\revb{0.3} &\revb{0.23}&\revb{0.07} \\
 

\bottomrule
\end{tabular}
\vspace{-4mm}
\end{table*}





\vspace{1mm}
\paratitle{Default parameters} Unless otherwise specified, the threshold of the Apriori algorithm is set to 0.1. 
For the SO dataset, the coverage thresholds are set to 0.5. 
The threshold for the SP and BFL fairness constraint is set at \$10k. For the German dataset, the coverage thresholds are set at
30\% and the fairness thresholds are set at 0.1. This configuration allows for the generation of multiple rules.

\vspace{1mm}
The results are shown in \cref{tab:problem_variants}, illustrating the trade-off between utility, coverage, and fairness. Without constraints, the expected utility is substantially higher, but this comes at the expense of greater disparities between protected and non-protected groups (as indicated by the unfairness score --- the difference between the expected utility of protected and non-protected). In the examined scenarios, coverage for both groups was achieved without constraints, but other protected group definitions may require them.


\paratitle{Stack Overflow}
Observe that while the expected utility for both protected and non-protected groups reaches its highest value in the no-constraints variant, the unfairness score is very high. This indicates that achieving SP fairness requires compromising on the expected utility for both protected and non-protected groups. Interestingly, rule coverage and individual fairness are difficult to achieve, as most rules fail to meet these criteria. This leads to lower expected utility for all groups. On the other hand, group coverage and fairness constraints are easier to satisfy, as they offer more flexibility by allowing the selection of some unfair rules alongside those specifically designed for the protected group.

\vspace{-1mm}
\begin{ruleset}{\textbf{3 Selected Rules out of 11 for SO (SP group fairness)}}
\small
    $\triangleright (\mathbf{S1_a})$ For individuals aged 24-34, pursue an undergraduate major in CS (exp utility protected: \textcolor{red}{10,292}, exp utility non-protected: \textcolor{blue}{22,586}).\\
        $\triangleright (\mathbf{S1_b})$ \common{For individuals with 6-8 years of coding experience, work with a computer 9 - 12 hours a day.} (exp utility protected: \textcolor{red}{17,161}, expe utility non-protected: \textcolor{blue}{19,254}).\\
$\triangleright (\mathbf{S1_c})$ \common{For males whose parents have a secondary school education, work as back-end developers} (exp utility protected: \textcolor{red}{51,542}, exp utility non-protected: \textcolor{blue}{46,354}).  
\end{ruleset}
\vspace{-1mm}

We show above the three example rules selected under group fairness constraint. The first rule $S1_a$ is more advantageous for the non-protected group, the second ($S1_b$) benefits both protected and non-protected groups similarly, while the third rule ($S1_c$) is more beneficial for the protected group. Overall, all these rules together satisfy the group fairness requirement. We also present three example rules selected under individual fairness constraints. In this case, all rules ($S2_a, S2_b, S2_c$) are nearly equally beneficial for both groups, but the overall expected utility is lower. Finally, consider the three example rules selected with no constraints. Here, all rules ($S3_a, S3_b, S3_c$ in the figure below) favor the non-protected group, highlighting the importance of including fairness constraints. 

\begin{ruleset}{\textbf{3 Selected Rules out of 20 for SO (SP individual fairness)}}
\small
    $\triangleright (\mathbf{S2_a})$ \common{For males aged 25-34 with no dependents, pursue a bachelor's degree} (exp utility protected: \textcolor{red}{16,158}, exp utility non-protected: \textcolor{blue}{18,134}).\\
        $\triangleright (\mathbf{S2_b})$ \common{For individuals aged 18 -24, work as back-end developers.} (exp utility protected: \textcolor{red}{12,664}, exp utility non-protected: \textcolor{blue}{14,101}).\\
$\triangleright (\mathbf{S2_c})$ \common{For individuals with dependents, pursue an undergraduate major in
CS} (exp utility protected: \textcolor{red}{16,124}, exp utility non-protected: \textcolor{blue}{17,138}).
\end{ruleset}

\begin{ruleset}{\textbf{3 Selected Rules out of 20 for SO (no fairness constraints)}}
\small
    $\triangleright (\mathbf{S3_a})$ \common{For White aged 25-34 with dependents, work with computer 9 - 12 hours a day and work as back-end developers} (exp utility protected: \textcolor{red}{11,147}, exp utility non-protected: \textcolor{blue}{32,248}).\\
        $\triangleright (\mathbf{S3_b})$ \common{For males aged 35-44 with dependents, work as back-end developers}. (exp utility protected: \textcolor{red}{11,189}, exp utility non-protected: \textcolor{blue}{40,207}).\\
$\triangleright (\mathbf{S3_c})$ \common{For students, pursue an undergraduate major in
CS} (exp utility for protected: \textcolor{red}{12,126}, exp utility for non-protected: \textcolor{blue}{22,174}).
\end{ruleset}


\paratitle{German}
While the expected utility for both protected and non-protected peaks in the no-constraints variant, the unfairness score is relatively high. This suggests that achieving BGL fairness necessitates compromising utility for both groups. Notably, to reduce the unfairness, it is feasible to impose either a rule coverage constraint or a rule coverage constraint combined with group fairness. 
We show three rules selected under BGL group fairness constraints below.
Since we are focusing on BGL fairness, which considers only the minimal gain for the protected group without regard for the gains of the non-protected group, we still observe a disparity between the two, even with a fairness constraint in place. 

\begin{ruleset}{\textbf{3 Selected Rules out of 20 for German (group BGL fairness)}}
\small
    $\triangleright (\mathbf{G1_a})$ For people aged 24-30 with 0-2 dependents, maintain a minimum balance of 200 DM in the checking account and pursue skilled employment
    (exp utility protected: \textcolor{red}{0.26}, exp utility non-protected: \textcolor{blue}{0.35}).\\
        $\triangleright (\mathbf{G1_b})$ For people seeking a loan to purchase furniture or equipment, maintain a minimum balance of 200 DM in the checking account (exp utility protected: \textcolor{red}{0.38}, exp utility non-protected: \textcolor{blue}{0.29}).\\
        $\triangleright (\mathbf{G1_c})$ For people seeking a loan for an unspecified purpose, maintain a minimum balance of 200 DM in the checking account and own a house. (exp utility protected: \textcolor{red}{0.54}, exp utility non-protected: \textcolor{blue}{0.41}).
\end{ruleset}


\begin{table*}[h!]
\centering
\small
\caption{Comparison of Solutions in Terms of Fairness}
\label{tab:fairness_variants}
\begin{tabular}{lccccccc}
\toprule
\textbf{Stack Overflow (SP fairness)} & \textbf{\# rules} & \textbf{coverage} & \textbf{coverage pro} & \textbf{exp utility} & \textbf{exp utility non-pro}&\textbf{exp utility pro} &\textbf{unfairness} \\
\midrule 
Group SP (2.5K)  &4& 97.82\%& 99.0\%& 20973.55& 20772.77& 18275.44& \textbf{2497.33} \\
Group SP (5K)  & 7& 97.31\%& 98.24\%& 22805.52& 23069.98& 18071.12& 4998.86\\
Group fairness (10K)  & 8& 97.52\%& 97.81\%& 27870.77& 27998.47& 17998.66& 9999.81 \\
Group SP (20K) & 20& \textbf{99.88\%}& \textbf{99.94\%}& \textbf{32671.11}& \textbf{32664.45}& \textbf{18423.64}& 14240.81 \\
\hline
Individual SP (2.5K) & 20& 99.95\%& 99.98\%& 24070.94& 24433.55& 12784.62& 11648.93 \\
Individual SP (5K)  & 20& \textbf{99.99\%}& \textbf{99.99\%}& 25526.1& 25911.22& 15327.21& \textbf{10584.01}\\
Individual SP(10K)   & 20& \textbf{99.99}\%& \textbf{99.99}\%& 28014.58& 28256.35& 14241.07& 14015.28 \\
Individual SP (20K) & 20& 99.51\%& 99.63\%& \textbf{29984.0}& \textbf{29966.29}& 14929.7& 15036.59\\
\midrule


\end{tabular}
\end{table*}
\vspace{-2mm}
\section{Experimental Evaluation}
\label{sec:experiments}
We present an experimental evaluation that evaluates \sysName\ effectiveness and efficiency. We aim to address the following questions:  
$\mathbf{Q1}$: How does the quality of our generated rulesets compare to that of existing methods? $\mathbf{Q2:}$ What is the efficiency of \sysName\ and how is it affected by various data and system parameters?

\subsection{Experimental Setup}
\label{sec:exp_setup}
\sysName\ was implemented in Python, and is publicly available in~\cite{fullversion}. 
CATE values computation was performed using the DoWhy library~\cite{dowhypaper}. The generated rules were translated into natural language using \reva{simple, manually constructed templates}.
We perform experiments on CloudLab ~\cite{Duplyakin+:ATC19} xl170 machines (10-core 2.4 GHz CPU, 64 GB RAM).
The datasets, protected groups, and default parameters considered are the same as those described in \cref{sec:casestudy}.



\vspace{1mm}
\paratitle{Baselines}
We compare \sysName\ with the following baselines:
\textbf{1. CauSumX}:
CauSumX \cite{DBLP:journals/pacmmod/YoungmannCGR24} is designed to find a summarized causal explanation for group-by-avg SQL query results. When applied directly to the datasets, it can be viewed as a solution to our problem with only an overall coverage constraint.
\textbf{2.IDS}~\cite{lakkaraju2016interpretable} is a framework for generating Interpretable Decision Sets for prediction tasks. IDS incorporates parameters restricting the percentage of uncovered tuples and the number of rules. These parameters were assigned the same values in our system.
\textbf{3. FRL}: The authors of \cite{chen2018optimization} proposed a framework for creating Falling Rule Lists (FRLs) as a probabilistic classification model. FRLs comprise if-then rules with antecedents in the if-clauses and probabilities of the desired outcome in the then-clauses, ordered based on associated probabilities.

\smallskip
Since IDS and FRL assume a binary outcome, we binned the salary variable in SO using the average value. To address fairness considerations, we run the baseline algorithms twice (excluding Brute-Force): Once on the entire dataset to obtain a set of rules applicable to the entire population, and again solely on the tuples belonging to the protected population to generate rules specifically tailored for them. \revb{We report the number of rules generated by the baselines, their coverage, and runtime. To compare the expected utility, we proceed as follows: The rules generated by IDS and FRL are prediction rules (e.g., IF owning a house = YES, THEN credit score = 1). As such, these rules do not provide an intervention to improve outcomes. We, therefore, treat the IF clauses
in two manners: (1) IF clauses as the selected grouping patterns and then apply step 2 (\cref{subsec:treatment_patterns}) of \sysName\ to determine the intervention patterns; (2) IF clauses as the selected intervention patterns, where the grouping pattern is the entire data. }





\subsection{Quality Evaluation (Q1)}
\label{exp:quality}
We compare the set of rules chosen by each baseline and \sysName. 

\begin{figure}[t]
    \centering
    \vspace{-3mm}\includegraphics[width=0.46\textwidth]{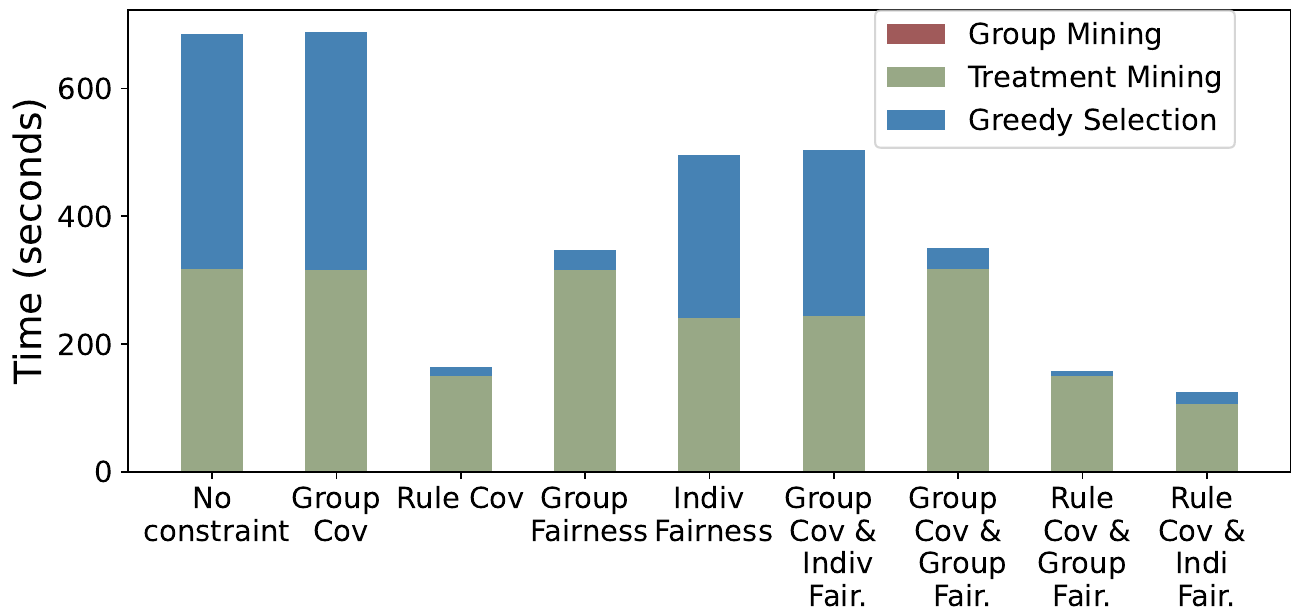}
    \caption{Runtime by-step of the \sysName\ algorithm (SO)}
    \label{fig:runtime_by_step}
\end{figure}

\paratitle{Stack Overflow} As discussed in \cref{sec:casestudy}, prescription rules selected without fairness constraints, similar to the behavior of CauSumX, were significantly more advantageous for non-protected.  
The rules generated by IDS do not suggest interventions to improve outcomes. For example, one rule states that if Country = Turkey and Age = 18-24 years, then the expected salary is low (with the outcome binned). Another key distinction is that these rules are not causal, as they are based on correlations in the data. For example, one rule indicates that if the years coding = 0-2 and Sexual Orientation = Gay or Lesbian, then the expected salary is low. Similarly, rules generated by FRL do not propose interventions to improve outcomes and are not causal. For example, one rule states that if Country = US and Sexual Orientation = Straight or Heterosexual, then the expected salary is high. 
In contrast, \sysName\ generates interventions aimed at improving the outcome by leveraging causal relationships. It also allows users to impose fairness constraints, ensuring that the protected group benefits from these interventions.




\paratitle{German}
Here again, with no fairness constraint (akin to CauSumX), the selected rules were mostly beneficial for the non-protected. 
Here again, the rules generated by IDS are not causal and do not offer an intervention. For example, one of the rules suggested that single females at the age of 35-41 are unlikely to get a loan.  
As before, the rules generated by FRL are also not causal and do not propose ways to improve the credit risk score. For example, one rule suggests that if a person has lived in a house for 4-7 years, their credit risk score is likely to be high. Another rule states that if the purpose of the loan is to buy a used car, the credit risk score is also likely to be high. Clearly, these rules rely on correlations in the data rather than causal relationships.
In contrast, \sysName\ generated a ruleset that offers interventions to improve the credit risk score based on causal relationships. Example selected rules are shown in \cref{sec:casestudy}.

\vspace{1mm}
\revb{We report the solution size, coverage, expected utility for protected and non-protected, and the unfairness of the rulesets generated using IDS and FRL (as explained in \cref{sec:exp_setup}).
The results are presented in 
{\bf \cref{tab:problem_variants}}. Notably, the expected utility for both protected and non-protected groups across both datasets is generally lower than that achieved by \sysName. \sysName\ consistently delivers higher expected utility for both groups and a smaller difference between these values. This indicates that our approach to mining grouping and intervention patterns is more effective than relying on these algorithms for the same purpose.  However, we note that the rules in IDS and FRL had different objectives (prediction accuracy) and had to be adapted for quantitative comparison using our measures.} 


\subsubsection{\reva{Robustness to the Causal DAG}}
\label{subsec:causal_DAG_robustness}
\reva{The quality of the generated rules may depend on the accuracy of the underlying causal DAG. To evaluate this, we examine the impact of different causal DAGs on the rules. The causal DAGs considered are as follows:
{
\textbf{(1) 1-layer Indep DAG:} A causal DAG where all attributes are independent of each other and only impact the outcome. This setting similar to the scenario where all the causal graph is ignored.
\textbf{(2) 2-layer Mutable DAG:} A simplified DAG where immutable attributes affect the mutable attributes, which impact the outcome variable. In this graph, all immutable attributes act as confounders but do not directly impact the outcome.
\textbf{(3) 2-layer DAG:} A simplified DAG where all variables affect the outcome but the mutable attributes are also confounded by all immutable attributes. }
\textbf{(4) PC DAG:} A causal DAG generated by the PC causal discovery algorithm~\cite{spirtes2001causation}}.

\reva{The results are depicted in \cref{tab:causal_dag_variants}. We report the expected utility as computed over the different causal DAGs. We observe that the expected utility remains similar for the Stack overflow dataset, demonstrating robustness towards the choice of causal dag. The results show some variability in German credit. However, the PC DAG and the original causal DAG are the most accurate (as they are based on the data distribution and domain knowledge) and achieve the highest coverage and expected utility.}

\begin{table*}[h!]
\centering
\small
\caption{\reva{Metrics Comparison with different Causal DAGs. 
}}
\label{tab:causal_dag_variants}
\begin{tabular}{p{40mm}ccccccc}
\toprule
\textbf{Stack Overflow (SP group fairness + group coverage)} & \textbf{\# rules} & \textbf{coverage} & \textbf{coverage pro} & \textbf{exp utility} & \textbf{exp utility non-pro}&\textbf{exp utility pro} &\textbf{unfairness} \\

\midrule 

Original causal DAG  & 11& 97.95\%& 98.85\%& 27934.76& 28144.58& 18145.23& 9999.35\\

\reva{1-Layer Indep DAG} &\reva{11}&\reva{98.38\%} & \reva{98.38\%}&\reva{28110.19}& \reva{28117}
&\reva{18117.45}
&\reva{9972} \\


\reva{2-Layer Mutable DAG} &\reva{10}	
&\reva{97.7\%}
 &\reva{98.4\%} &\reva{28198.59}&\reva{28193.09} &\reva{18193.23
}&\reva{9999.86} \\

\reva{2-Layer DAG} &\reva{10}	
&\reva{98.47\%}
 &\reva{98.87\%} &\reva{28106.4}&\reva{28211.17} &\reva{18211.4}&\reva{9999.77} \\

\reva{PC DAG} &\reva{10}&\reva{97.7\%} &\reva{98.4\%} &\reva{28198.59}&\reva{28193.09} &\reva{18193.23}&\reva{9999.86} \\

\midrule

\textbf{German Credit (BGL group fairness + group coverage)} & \textbf{\# rules} & \textbf{coverage} & \textbf{coverage pro} & \textbf{exp utility} & \textbf{exp utility non-pro}&\textbf{exp utility pro}&\textbf{unfairness}   \\
\midrule 

Original causal DAG  & 6& 100.0\%& 100.0\%& 0.36& 0.37& 0.31& 0.06 \\
\reva{1-Layer Indep DAG} &\reva{12}&\reva{100\%} &\reva{100\%} &\reva{0.31}& \reva{0.31}&\reva{0.29}&\reva{0.02} \\
\reva{2-Layer Mutable DAG}&\reva{13} &\reva{76.20\%}		
&\reva{79.35\%} & \reva{0.22}&\reva{0.22}&\reva{0.2} &\reva{0.02} \\

\reva{2-Layer DAG} &\reva{11}	
&\reva{71.20\%}
 &\reva{73.91\%} &\reva{0.26}&\reva{0.25} &\reva{0.23}&\reva{0.02} \\

\reva{PC DAG} &\reva{24}&\reva{100.00\%}	
 &\reva{100.00\%} &\reva{0.39}&\reva{0.39} &\reva{0.26}&\reva{0.13} \\
\bottomrule
\end{tabular}
\end{table*}

\subsection{Scalability Evaluation (Q2)}
\label{exp:scalability}


\begin{figure}[t]
    \begin{subfigure}[b]{0.46\textwidth}
        \centering
        \includegraphics[width=0.6\textwidth]{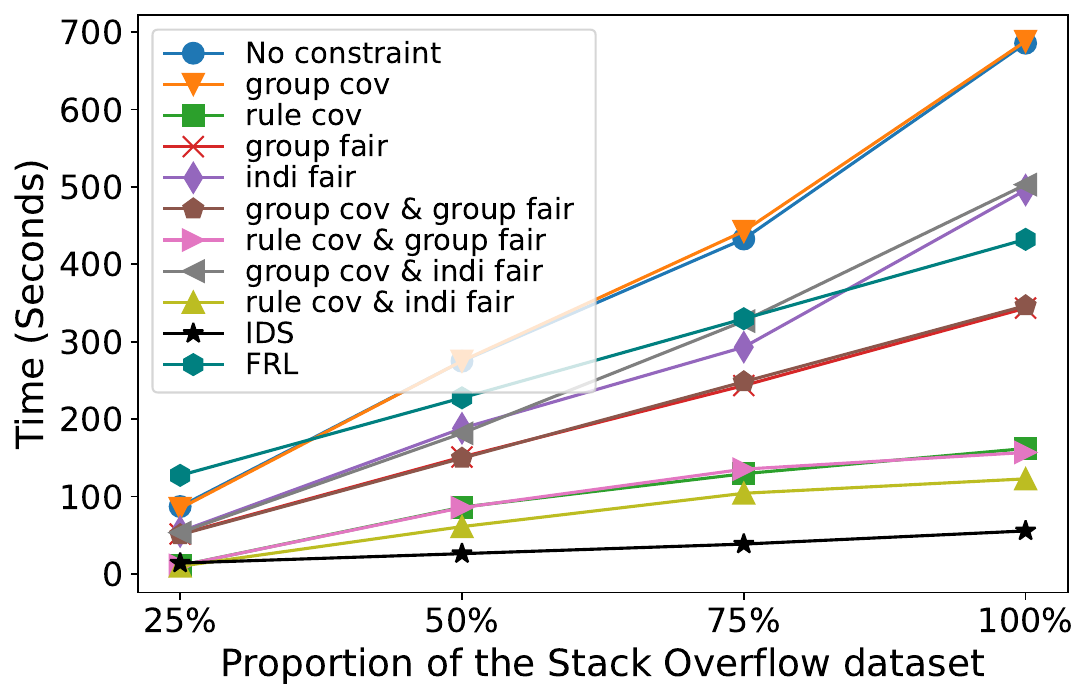}
    \end{subfigure}
    \caption{\revb{Runtime as a function of the dataset size (SO)}}
\label{fig:runtime_dataset_size}
\end{figure}

\begin{figure}[t]
    \vspace{-5mm}
    \centering
        \begin{subfigure}[b]{0.48\textwidth}
        \centering
        \includegraphics[width=\textwidth]{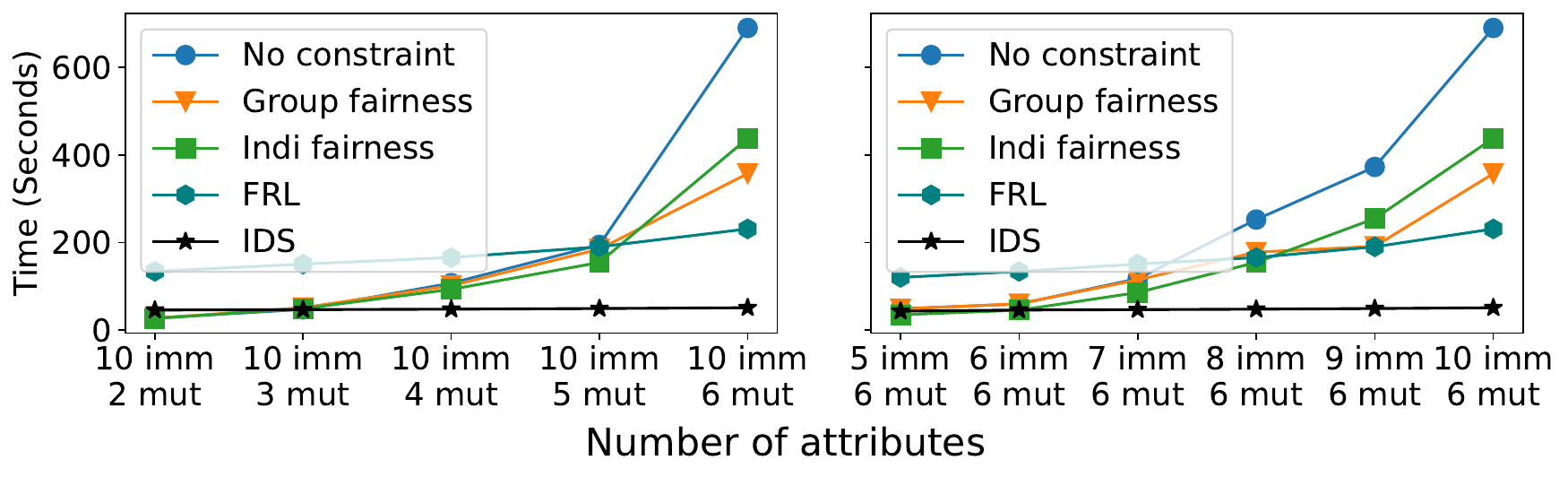}
        \end{subfigure}
    \caption{\revb{Runtime as a function of number of mutable and immutable attributes for SO with statistical parity}}
\label{fig:runtime_attributes}
\end{figure}


\paratitle{Breakdown analysis by step}
Figure~\ref{fig:runtime_by_step} shows the runtime comparison of \sysName for different problem settings. Observe that using rule coverage constraint has the lowest runtime because it helps to prune rules which do not satisfy the coverage constraint. Employing rule coverage with individual fairness is the fastest among all settings, while no constraint setting takes the longest time.
The time taken by the group mining phase is less than $2$ seconds across all setups, and is therefore not visible in the plot. The intervention mining phase (Step 2) is the most inefficient phase, which takes around $6$ mins for the unconstrained setting. The running time of these components aligns with our time complexity analysis (\cref{sec:algo}). Due to space restrictions, we do not present the corresponding plot for German dataset. All conclusions remain the same but the overall running time is $\approx 10\times$ faster due to its smaller size.

\revb{The running time of \sysName and the baselines is comparable. 
FRL is an order of magnitude slower than IDS because it uses a Bayesian modeling approach to simultaneously select a subset of rules and determine their optimal order, which involves solving a computationally intensive combinatorial problem. In contrast, IDS leverages submodular optimization on an unordered set of rules, significantly reducing the size of the search.}
We now analyze the impact of system parameters and data size on performance.

\smallskip
\paratitle{Data Size} \Cref{fig:runtime_dataset_size} compares the running time of \sysName\ \revb{and the baselines} for varying dataset sizes. We observe that the time taken by \sysName\  \revb{and the baselines} increases linearly for most of the settings, \revb{with \sysName\ demonstrating a runtime comparable to IDS under certain configurations}. We also observed that the quality of rules returned by sampling $25\%$ of the data points is comparable with the rules returned by using the whole dataset. Therefore, sampling-based optimizations can help to reduce the running time from $11$ min to less than $2$ min for the unconstrained setting and less than a minute with fairness constraints. 


\smallskip
\paratitle{Number of Attributes} Figure~\ref{fig:runtime_attributes} shows the runtime of \sysName\ while increasing the number of mutable and immutable attributes. 
On increasing the number of mutable attributes, the number of intervention patterns increases exponentially while on increasing immutable attributes, the number of grouping patterns increases exponentially. Therefore, both have a similar impact on runtime. \revb{IDS and FRL do not distinguish between mutable and immutable attributes and there the runtime increases slightly due to an increase in the number of attributes, as more rules are considered.}

\reva{In the following, we omit the results for the IDS and FRL baselines, as these parameters do not impact their runtime.}

\smallskip
\paratitle{Fairness Threshold} 
\Cref{tab:fairness_variants} presents the results for varying $\epsilon$ for group and individual fairness. We observe that the unfairness of the returned solution increases with the increase in $\epsilon$.  Additionally, the overall expected utility increases but the expected utility of the protected individuals decreases. This result matches our intuition as highly unfair rules are selected for higher values of $\epsilon$. We also notice that the greedy algorithm satisfies the group fairness constraint in all scenarios (unfairness is always less than the desired threshold).

For individual fairness, the overall utility increases monotonically with $\epsilon$. However, the rate of growth for individual fairness is slower than that of group fairness.
One interesting observation about individual fairness is that when all rules have statistical parity difference less than $2500$, the overall unfairness is still around $11K$. This sudden increase in unfairness when considering multiple fair rules together is because we evaluate the upper bound of unfairness by taking the difference between max utility of unprotected and min utility of protected individuals. On manual inspection, we observed that all rules are indeed individually fair.



\smallskip
\paratitle{Coverage Threshold} With the change in coverage thresholds, we do not observe major difference in the overall results because the majority of the rules exhibit very high coverage (\cref{tab:problem_variants}). 

\smallskip
\paratitle{Apriori Threshold} 
We observe that increasing the Apriori threshold $\tau$ leads to a reduction in the number of grouping patterns considered, and thus to a decrease in runtime. However, our findings indicate that higher $\tau$ values lead to a decrease in both utility and fairness. Based on our findings, we recommend using a default value of $0.1$, which provides satisfactory results in terms of coverage, utility, fairness and runtime.





\section{Limitations and Future Work}
\label{sec:conc}

\sysName\ generates actionable, causal-based recommendations to improve a target outcome while incorporating coverage and fairness constraints. \common{To the best of our knowledge, this is one of the first works in this direction, and several directions of future work remain. In this section, we discuss some of the current limitations of \sysName\ and future directions. }
\par
\reva{
{\bf Generation and usage of rules by \sysName.~} \sysName\ can be used to recommend actions for different subpopulations toward optimizing a target. As an example scenario, a policymaker may select the target outcome and the parameters for coverage and fairness constraints (which may be iteratively varied based on the application).
\sysName\ then generates a prescription ruleset as recommended actions for different subpopulations.
The current framework assumes that the policymaker is trustworthy, will not misuse the rules, and will publish the relevant recommendations for each subpopulation. However, it is important to note that if not all rules are provided to all subpopulations, disparities among subpopulations may increase.
In addition, the generated rules 
may not impact all individuals receiving the recommendation in the same way. The gain in objective may vary across different subpopulations. For example, an increase in $\$10k$ revenue may have varied impacts in different countries, depending on the cost of living and purchasing power. Addressing these will be interesting future work.
}

\reva{{\bf Considering constraints,  costs, and resources in rule generation.~}} 
The current framework does not account for the cost of interventions. Some interventions may be impractical (e.g., pursuing a bachelor’s degree in CS for someone who already holds a PhD in CS) or vary significantly in cost (e.g., moving to the US versus learning Python). 
\reva{Further, the generated rules do not consider global constraints, e.g., if the targeted outcome is the salary in an institution, there may be a budget. 
Future research will incorporate intervention costs to generate budget-constrained rules and address finite resource allocation scenarios to account for cases where the population size that can achieve improved outcomes is limited}.

{\bf Extension to multi-table data.~} \sysName\ currently supports a single-relation database without dependencies among tuples to ensure compliance with the SUTVA assumption~\cite{rubin2005causal} (discussed in \cref{sec:background-causal}). However, this assumption breaks down even in single-table databases with tuple dependencies. In single-table settings, intervention and grouping patterns are straightforward to define. Extending these definitions to multi-table databases, where grouping attributes and interventions may originate from different tables, introduces a significant challenge. This complexity arises due to many-to-many relationships and cross-table patterns. 
Previous work, such as \cite{salimi2020causal,galhotra2022hyper}, has extended causal models to handle multi-table data, but they have not explicitly targeted recommendations for subgroups. Expanding our framework to support multi-relational databases with complex dependencies remains an important direction for future research. Notably, prior work leveraging causal inference \cite{ma2023xinsight, youngmann2022explaining, salimi2018bias, DBLP:journals/pacmmod/YoungmannCGR24} has also primarily focused on single-table settings.


{\bf Robustness of rules.~} The generated rules may be influenced by several factors, including the method used to evaluate causal effects, the thresholds set for the constraints, the overall quality of the data, and the quality of the causal DAG.
\reva{In this work, we assume that the causal DAG is provided as part of the input, with the responsibility for validating its correctness resting on the policymaker. Nonetheless, the causal DAG only needs to specify causal dependencies between variables without detailing the nature of those dependencies. Developing methods that are robust to inaccuracies in the DAG is an important direction for future work. }

\common{{\bf Explainability and prescriptive causal nature of rules.~} While if-then rules for prediction or causation are considered explainable or interpretable in the literature \cite{lakkaraju2016interpretable,pradhan2022interpretable,van2021evaluating,guidotti2018local,chen2018optimization}, we note that no additional explanations or justifications come with the rules mined by \sysName. 
Generating meaningful explanations to describe how the rules impact the outcome and the variability of the outcome within various sub-populations is deferred to future work. }

\smallskip
\common{To conclude, observational causal analysis is the main foundation for any {\em prescription} or {\em recommendation} beyond predictions, when a randomized controlled trial is not possible due to cost, ethics, or feasibility issues. However, the analysis depends on assumptions (ignorability, causal DAG) that may not hold in a scenario and one should know the assumptions and limitations of these claims. How the rules should be used in practice considering practical and fairness aspects is a general direction of future work.}

\begin{acks}
 This work was partially supported by the NSF awards IIS-2008107 and IIS-2147061, and a grant from Infosys. Additional funding was provided by the Henry and Marilyn Taub faculty for computer science at the Technion.
\end{acks}

\small
\bibliographystyle{ACM-Reference-Format}
\balance
\bibliography{vldb_sample}

\section{Missing Proofs}

\begin{proof}[Proof of Lemma \ref{lemma:individual_rules}]
The utility of a rule $r$ denotes the expected increase in outcome $O$ when all individuals within the subgroup $\pattern_g$ are treated with $\pattern_t$. 
\begin{eqnarray*}
    utility(r) &=& \frac{1}{|\pattern_g|} \sum_{i\in coverage(\pattern_g)} utility_i(\pattern_t)
\end{eqnarray*}
where $utility_i(\pattern_t)$ denotes the utility for tuple $i$ with respect to treatment $\pattern_t$. Since utility(r) is an average over multiple different utilities, the utility will be higher than the expected value for certain tuples in $coverage(\pattern_g)$.
Let $i^* = \arg \max utility_i(\pattern_t)$.

Consider a new prescription rule $r' (i^*,\pattern_t)$ which considers the same treatment $\pattern_t$ for the tuple $i^*$. Therefore,
    $utility(r') = utility_i(\pattern_t) > utility(r)$.
\end{proof}

\subsection{Hardness Results}
We next study the complexity of the \probName\ problem under different constraint combinations. 
We show that \probName\ is equivalent to optimizing a non-negative and monotone submodular function. 
Furthermore, the individual fairness constraint and rule coverage constraints are matroid constraints. Therefore, a greedy algorithm is appropriate approach to solve the problem.

\begin{proposition}
\label{prop:unconstrained}
    The optimization objective of \probName\ problem is a non-negative submodular function.
\end{proposition}

\begin{proof}[Proof of \cref{prop:unconstrained}]
 According to \cite{lakkaraju2016interpretable}, the size objective is a non-negative and submodular function. Similarly, the expected utility—assuming each individual receives a single rule and selects the best option—is also a non-negative submodular function. As a result, the linear combination of these functions remains a non-negative submodular function, and maximizing it is known to be NP-hard \cite{khuller1999budgeted}.
\end{proof}

\begin{proposition}
\label{prop:matroid}
    Individual fairness and rule coverage constraints are matroid constraint
\end{proposition}

\begin{proof}[Proof of Proposition \ref{prop:matroid}]

We will show these constraints satisfy the following properties: 

\begin{enumerate}
    \item \textbf{Hereditary Property}: If  $S$ is an independent set, then every subset of $S$ is also an independent set.
    \item \textbf{Exchange Property}: If $S$ and $T$ are independent sets and $|S| < |T|$, then there exists an element $e \in T \setminus S$ such that $S \cup \{e\}$  is also an independent set.
\end{enumerate}
These two properties ensure that the set system behaves like a matroid.

In our setting Start by specifying the ground set is all possible rules and what qualifies as an "independent set" is a subset of rules satisfying a constraint. 

\paragraph*{\bf Individual Fairness} 
If a set of rules $R$ satisfies the individual fairness constraint, this means each rule within $R$ individually satisfies the constraint. Consequently, any subset $R' \subseteq R$ also upholds individual fairness. This further implies the exchange property, as any rule that satisfies individual fairness can be added to an individually fair set of rules while preserving individual fairness.

\paragraph*{\bf Rule coverage} 
If a set of rules $R$ satisfies the rule coverage constraint, this means each rule within $R$ individually satisfies the rule coverage constraint. Consequently, any subset $R' \subseteq R$ also satisfies the rule coverage constraint. This also implies the exchange property, as any rule that satisfies rule coverage can be added to a set of rules satisfying rule coverage while preserving the rule coverage constraints.
\end{proof}

For the group coverage, we can show that merely finding a solution that satisfies the constraints, even without maximizing expected utility, is NP-hard via a reduction from the Set Cover problem~\cite{feige1998threshold}. 

\begin{proposition}
    \label{prop:group_coverage}
    \probName\ with a group-coverage constraint is NP-hard
\end{proposition}

\begin{proof}[Proof of \ref{prop:group_coverage}]
In the decision version of the Set Cover problem, we are given a universe of elements $U = \{x_1, \ldots, x_{n'}\}$, a collection of $m$ subsets $S_1, \ldots, S_{m'} \subseteq U$ and a number $k$. The question is whether there exists a cover of $U$ of at most $k'$ subsets.

In the decision version of \probName\, 
we are searching for a set of rules $R$ such that: $f(R) \geq \tau$, where:
$$f(R) = \lambda_1 {\cdot} (l - size(R)) {+} \lambda_2 {\cdot} ExpUtility(R)$$
 such that $R$ satisfies the group-coverage constraints, defined by the parameter $\theta$. In this proof, we assume no protected group is given, namely the constraint requires that the selected ruleset $R$ would cover at least a $\theta$ fraction of the population (i.e., the protected group to be the empty group).

Given an instance of the set cover problem, we build an instance of the \probName\ problem as follows. 
We build a relation $R$ with $m'+1$ attributes, $\attrset = (A_1, \ldots, A_{m'}, O)$, and containing $n'+m'$ tuples. For each element $x_i \in U$, we create a tuple $t_i$, such that $t_i[A_j]=1$ iff $x_i \in S_j$. 
We further add $m'$ tuples $t_{S_j}$ such that $t_{S_j}[A_j] = 1$, $t_{S_j}[O] = 0$, and $t_{S_j}[A_p] = l \neq 0$ for all $p \neq j$ where $l$ is a unique number not used anywhere else in an attribute of $R$. We set the outcome variable to be $O$.

Here, $\pattern_g$ can be any predicate.  
Note that each set of tuples defined by a pattern can only have an outcome of $0$, as the outcome of all tuples is $0$. 
Therefore, the utility of all intervention patterns is $0$. 
For \probName, we further define the threshold for the group coverage constraint $\theta = \frac{n'+k'}{n'+m'}$.
The underlying causal DAG, $G$, only contains the edges of the form $A_j \to O$ for all $1\leq j \leq m'$. 
We claim that there exists a cover of $U$ with at most $k$ sets iff there exists a solution $R$ to \probName\ such that $f(R) \geq (l-k)$. 

($\Rightarrow$) Assume we have a collection $S_{j_1}, \ldots, S_{j_k}$ such that $\cup_{j = j_1}^{j_k} S_j = U$. We show that there is a solution for \probName\ as follows. For each $S_{j_1}$, we choose for the solution the pattern $\pattern_g^{j_i}: A_{j_i} = 1$. 
We show that $R = \{(\pattern_g^{j_1}, \emptyset), \ldots, (\pattern_g^{j_k}, \emptyset)\}$ is a solution to \probName.
The intervention pattern can be any pattern, as the utility of every intervention pattern is $0$.
First, we note that all tuples of the form $t_i$ are covered by at least one grouping pattern by their definition. For the $m$ remaining tuples, we have coverage of at most $k$ tuples. These are the tuples $t_{S_{j_i}}$ that have $A_{j_i} = 1$. 
Thus, the number of covered tuples is exactly $n'+k'$ out of $n'+m'$ tuples in $R$. 
If there are fewer than $k$ tuples we can augment the original cover with arbitrary sets to obtain a cover of size $k$.

($\Leftarrow$) Assume we have a solution $R$ to \probName\ with the aforementioned parameters. We show that we can find a solution to the set cover problem. 
First, note that since $f(R) \geq (l-k)$ and the expected utility is always $0$, that means we have selected no more than $k$ rules.

Suppose $R = \{(\pattern_g^{j_1}, \emptyset), \ldots, (\pattern_g^{j_k}, \emptyset)\}$. 
We first claim that no grouping pattern that includes $A_i = 0$ in a conjunction can be included in $R$ as such a pattern will not cover any tuple $t_{S_j}$ since these tuples do not have an attribute with value $0$ by definition (and any other number other than $1$ will only cover a single tuple). Thus, the number of covered tuples will be $< \frac{n'+k'}{n'+m'} = \theta$, which would contradict the assumption that this is a valid solution to \probName\ that satisfies the coverage constraint. 
Hence, all patterns are of conjunctions of $A_i = 1$. For each intervention pattern of the form $\pattern_g = \wedge_{j=i_a}^{i_b} (A_j = 1) \land (A_p = l)$, we choose an arbitrary attribute in the conjunction $A_{j}$ if $(A_j = 1) \in \pattern_g$ and choose $S_{j}$ for the cover. Finally, if there is an uncovered element $x$ in $U$ and $R$ includes a pattern in of the form $\pattern_g = (A_j = l)$ where $l \neq 1$, we choose for the cover a set $S$ that covers $x$ arbitrarily.  
We claim that the chosen collection of sets is a cover of $U$. 
To see this, 
recall that we claimed that the coverage of $R$ is at least $n+k$. 
If the coverage includes tuples of the form $t_{S_j}$, then each pattern covers a single tuple. Suppose these patterns are $\pattern_a, \ldots, \pattern_b$. 
When building the coverage, instead of these patterns, we add a set that covers elements not yet covered by existing patterns. Thus, there are at least $b-a$ covered elements from $U$ in addition to the $n+k-(b-a)$ tuples covered by the patterns. Thus, the set cover we have assembled contains $n-(b-a) + (b-a) = n$ elements and covers all elements in $U$. 
\end{proof}

\end{document}